\documentclass[a4paper,USenglish]{lipics-v2016-noArt}
\usepackage[utf8]{inputenc}
\usepackage{booktabs}
        
\usepackage{microtype}
\usepackage{amsmath,amsfonts,amssymb,amsthm,bm,bbm, authblk, graphicx}
\usepackage{thmtools, thm-restate}
\theoremstyle{plain}
\numberwithin{theorem}{section}

\DeclareRobustCommand{\PLBU}{PLB\nobreakdash-U\xspace}
\DeclareRobustCommand{\PLBL}{PLB\nobreakdash-L\xspace}
\DeclareRobustCommand{\PLBN}{PLB\nobreakdash-N\xspace}
\DeclareRobustCommand{\PLBUN}{PLB\nobreakdash-(U,N)\xspace}
\DeclareRobustCommand{\PLBUL}{PLB\nobreakdash-(U,L)\xspace}
\DeclareRobustCommand{\PLBULN}{PLB\nobreakdash-(U,L,N)\xspace}
\usepackage[labelfont=bf,font=small,indention=0cm,margin=0cm]{caption} 
\usepackage{booktabs}
\usepackage{xspace}
\usepackage{url}
\usepackage{float}
\floatstyle{ruled}
\newfloat{myalgorithm}{thp}{lop}
\floatname{myalgorithm}{Algorithm}
\usepackage{enumerate}
\usepackage{algorithm}
\usepackage{algpseudocode}
\usepackage[numbers,sort&compress,longnamesfirst,sectionbib]{natbib}
\usepackage{setspace,color}   

\usepackage{complexity}

\renewcommand{\R}{\mathbb{R}}
\newcommand{\N}{\mathbb{N}}
\newcommand{\Oh}{\mathcal{O}}
\newcommand{\Vol}{\textsc{vol}}
\newcommand{\opt}{\textsc{opt}}

\newcommand{\whp}{w.\thinspace h.\thinspace p.\xspace}
\def\argmax{\operatornamewithlimits{argmax}}

\def\min{\operatorname{min}}
\def\max{\operatorname{max}}
\def\deg{\operatorname{deg}}

\def\Ex{\mathbb{E}}
\def\Pr{\operatorname{Pr}}

\renewcommand{\algref}[1]{Algorithm~\ref{alg:#1}}
\newcommand{\defref}[1]{Definition~\ref{def:#1}}
\newcommand{\thmref}[1]{Theorem~\ref{thm:#1}}
\newcommand{\lemref}[1]{Lemma~\ref{lem:#1}}
\newcommand{\corref}[1]{Corollary~\ref{cor:#1}}
\newcommand{\secref}[1]{Section~\ref{sec:#1}}
\newcommand{\eq}[1]{equation~\eqref{eq:#1}}
\newcommand{\ineq}[1]{inequality~\eqref{eq:#1}}
\renewcommand{\epsilon}{\ensuremath{\varepsilon}}
\let\oldsqrt\sqrt
\def\hksqrt{\mathpalette\DHLhksqrt}
\def\DHLhksqrt#1#2{\setbox0=\hbox{$#1\oldsqrt{#2\,}$}\dimen0=\ht0
   \advance\dimen0-0.2\ht0
   \setbox2=\hbox{\vrule height\ht0 depth -\dimen0}   {\box0\lower0.4pt\box2}}
\renewcommand\sqrt\hksqrt
\renewcommand{\leq}{\leqslant}
\renewcommand{\geq}{\geqslant}
\renewcommand{\le}{\leqslant}
\renewcommand{\ge}{\geqslant}
\catcode`@=11
\def\nphantom{\v@true\h@true\nph@nt}
\def\nvphantom{\v@true\h@false\nph@nt}
\def\nhphantom{\v@false\h@true\nph@nt}
\def\nph@nt{\ifmmode\def\next{\mathpalette\nmathph@nt}  \else\let\next\nmakeph@nt\fi\next}
\def\nmakeph@nt#1{\setbox\z@\hbox{#1}\nfinph@nt}
\def\nmathph@nt#1#2{\setbox\z@\hbox{$\m@th#1{#2}$}\nfinph@nt}
\def\nfinph@nt{\setbox\tw@\null
  \ifv@ \ht\tw@\ht\z@ \dp\tw@\dp\z@\fi
  \ifh@ \wd\tw@-\wd\z@\fi \box\tw@}
\newcount\minute \newcount\hour \newcount\hourMins
\def\now{\minute=\time \hour=\time \divide \hour by 60 \hourMins=\hour \multiply\hourMins by 60
  \advance\minute by -\hourMins \zeroPadTwo{\the\hour}:\zeroPadTwo{\the\minute}}

\def\today{\the\year-\zeroPadTwo{\the\month}-\zeroPadTwo{\the\day}}
\def\zeroPadTwo#1{\ifnum #1<10 0\fi #1}
\DeclareRobustCommand{\PLBU}{PLB\nobreakdash-U\xspace}
\DeclareRobustCommand{\PLBL}{PLB\nobreakdash-L\xspace}
\DeclareRobustCommand{\PLBN}{PLB\nobreakdash-N\xspace}
\DeclareRobustCommand{\PLBUN}{PLB\nobreakdash-(U,N)\xspace}
\DeclareRobustCommand{\PLBUL}{PLB\nobreakdash-(U,L)\xspace}
\DeclareRobustCommand{\PLBULN}{PLB\nobreakdash-(U,L,N)\xspace}
\usepackage[labelfont=bf,font=small,indention=0cm,margin=0cm]{caption} 
\usepackage{booktabs}
\title{Greed is Good for Deterministic Scale-Free Networks}
\author{Ankit Chauhan}
\author{Tobias Friedrich}
\author{Ralf Rothenberger}
\affil{Hasso Plattner Institute, Potsdam, Germany}
\renewcommand\Authands{ and }
\authorrunning{Ankit\,Chauhan, Tobias\,Friedrich and Ralf\, Rothenberger }
\Copyright{Ankit Chauhan, Tobias Friedrich and Ralf Rothenberger}
\begin{document}

\maketitle

\begin{abstract}
Large real-world networks typically follow a power-law degree distribution.
To study such networks, numerous random graph models have been proposed.
However, real-world networks are not drawn at random. Therefore,
Brach, Cygan, {\L}{{a}}cki, and Sankowski [SODA 2016]
introduced two natural deterministic conditions:
(1) a power-law upper bound on the degree distribution (PLB-U) and
(2) power-law neighborhoods, that is, the degree distribution of neighbors of each vertex is also upper bounded by a power law (PLB-N).
They showed that many real-world networks
satisfy both deterministic properties and exploit them to design faster algorithms
for a number of classical graph problems.

\qquad We complement the work of Brach et al. by showing that some well-studied random graph models exhibit both the mentioned PLB properties and additionally also a power-law lower bound on the degree distribution (PLB-L). All three properties hold with high probability for Chung-Lu Random Graphs and Geometric Inhomogeneous Random Graphs and almost surely for Hyperbolic Random Graphs. As a consequence, all results of Brach et al. also hold with high probability or almost surely for those random graph classes.

\qquad In the second part of this work we study three classical $\NP$-hard combinatorial optimization problems
on  PLB networks. It is known that
on general graphs with maximum degree~$\Delta$,
a greedy algorithm, which chooses nodes in the order of their degree,
only achieves a
$\Omega(\ln \Delta)$-approximation for \textsc{Minimum Vertex Cover} and
\textsc{Minimum Dominating Set},
and a $\Omega(\Delta)$-approximation for \textsc{Maximum Independent Set}.
We prove that the PLB-U property suffices
for the greedy approach to achieve a constant-factor approximation 
for all three problems.
We also show that all three combinatorial optimization problems are \APX -complete even if all PLB-properties holds  hence, PTAS cannot be expected unless \P=\NP.
\end{abstract}

\newpage

\section{Introduction}

A wide range of real-world networks exhibit a degree distribution that resembles a power-law~\cite{albert2002statistical, Newman03}. This means that the number of vertices with degree~$k$ is proportional to $k^{-\beta}$, where $\beta>1$ is the power-law exponent, a constant intrinsic to the network.
This applies to Internet topologies~\cite{FaloutsosFF99}, the Web~\cite{KumarRRT99,BarabasiAlbert1999}, social networks~\cite{AdamicBA03}, power grids~\cite{Phadke09}, and literally hundreds of other domains~\cite{New03}.
Networks with a power-law degree distribution are also called scale-free networks and have been widely studied.

To capture the degree distribution and other properties of scale-free networks, a multitude of random graph models have been proposed.
These models include Preferential Attachment~\cite{BarabasiAlbert1999},
the Configuration Model~\cite{chunglumassive},
Chung-Lu Random Graphs~\cite{Chung:2002:connected} and
Hyperbolic Random Graphs~\cite{krioukov2010hyperbolic}.
Despite the multitude of random models, none of the models truly has the same set of properties
as real world networks.

This shortcoming of random graph models motivates studying deterministic properties of scale-free models, as these deterministic properties
can be checked for real-world networks.
To describe the properties of scale-free networks without the use of random graphs,~\citet{AielloCL00} define $(\alpha, \beta)$-Power Law Graphs.
The problem of this model is that it essentially demands a perfect power law degree distribution, whereas the degree distributions of real networks normally exhibit slight deviations from power-laws.
Therefore, $(\alpha, \beta)$-Power Law Graphs are too constrained and do not capture most real networks.

To allow for those deviations in the degree distribution~\citet{PLBNetworks} define buckets containing nodes of degrees $\left[2^i, 2^{i+1}\right)$.
If the number of nodes in each bucket is at most as high as for a power-law degree sequence, a network is said to be \emph{power-law bounded}, which we denote as a network with property \textbf{\PLBU}.
They also define the property of \emph{PLB neighborhoods}:
A network has PLB neighborhoods if every node of degree $k$ has at most as many neighbors of degree at least $k$ as if those neighbors were picked independently at random with probability proportional to their degree. This property we abbreviate as \textbf{\PLBN}. A formal definition of both properties can be found in \secref{prelim}.
\citet{PLBNetworks} showed experimentally that \PLBUN properties hold for many real-world networks, which implies that the mentioned problems on these networks can be solved even faster than the known worst case lower bound for the general graphs.

\medskip
\section{Our Contribution}
\paragraph*{PLB properties in power-law random graph models}
The \PLBUN properties are designed to describe power-law graphs in a way that allows analyzing algorithms deterministically. As already mentioned, there is a mutitude of random graph models~\cite{chunglumassive,Chung:2002:connected,BarabasiAlbert1999,krioukov2010hyperbolic}, which can be used to generate power-law graphs. \citet{PLBNetworks} proved that the Erased Configuration Model~\cite{chunglumassive} follows \PLBU and \whp also \PLBN. Since the Erased Configuration Model has a fixed degree sequence, it is relatively easy to prove the \PLBU property, but it is quite technical to prove the \PLBN property. There are other power-law random graph models, which are based on the expected degree sequence, e.g. Chung-Lu Random Graphs~\cite{Chung:2002:connected}. Brach et al. argued that for showing the \PLBU property on these models, a typical concentration statement does not work, as it accumulates the additive error for each bucket. They leave it as a \textit{challenging open question}, whether other random graph models also produce graphs with \PLBUN properties with high probability\footnote{We say that an event $E$ holds {\em \whp}, if there exists an $\delta > 0$ such that $\Pr[E] \geq 1- \Oh(n^{-\delta})$, and {\em almost surely} if it holds with probability $\Pr[E] \geq 1- o(1)$.}.

The models we consider in \secref{random} are Chung-Lu Random Graphs, Geometric Inhomogeneous Random Graphs and Hyperbolic Random Graphs.

\textbf{Chung-Lu Random Graphs satisfy \PLBUN:}
Chung-Lu Random Graphs~\cite{Chung:2002:connected} assume a sequence of expected degrees $w_1,\ w_2,\ldots,\ w_n$ and each edge $(i,j)$ exist independently at random with probability $\min(1,\frac{w_i\cdot w_j}{W})$, where $W=\sum_{i=1}^{n}{w_i}$.
We show the following theorem:
\begin{restatable*}{theorem}{statechunglu}\label{thm:chunglu}
Let $G$ a Chung-Lu random graph whose weight sequence $\vec{w}$ follows a general power law with exponent $\beta'>2$ and an $\eta$ with $\beta'-\eta>2$.
Then, \whp $G$ fulfills \PLBU and \PLBN with $\beta=\beta'-\eta$, $t=0$ and some constants $c_1$ and $c_2$.
\end{restatable*}

\medskip
\textbf{Hyperbolic Random Graphs satisfy \PLBUN:}
Hyperbolic Random Graphs~\cite{krioukov2010hyperbolic} assume an underlying hyperbolic space. 
Each node is positioned uniformly at random in this space and connected to all nodes in a certain maximal hyperbolic distance to it.
For Hyperbolic Random Graphs we show the following:
\begin{restatable*}{theorem}{statehyperbolic}\label{thm:hyperbolic}
Let $G$ be a hyperbolic random graph with $\alpha_H>\frac{1}{2}$.
Then,  $G$ almost surely fulfills \PLBU and \PLBN  with $\beta=2\alpha_H+1-\eta$, $t=0$, constant $\eta>0$ and some constants $c_1$ and~$c_2$.\end{restatable*}

\medskip
\textbf{Geometric Inhomogeneous Random Graphs satisfy \PLBUN:}
Geometric Inhomogeneous Random Graphs (GIRGs)~\cite{BKL15,KochL16,BKL16} consider expected degree vector and an underlying geometry.

In GIRGs, all nodes draw a position uniformly at random and each edge $(i,j)$ exist independently with a probability depending on $\frac{w_i\cdot w_j}{W}$ and the distance of $i$ and $j$ in the underlying geometry.
We show:
\begin{restatable*}{theorem}{stategirg}\label{thm:girg}
Let $G$ be a $ GIRG$ whose weight sequence $\vec{w}$ follows a general power-law with exponent $\beta'>2$ and an $\eta$ with $\beta'-\eta>2$.
Then, \whp $G$ fulfills \PLBU and \PLBN with $\beta=\beta'-\eta$, $t=0$ and some constants $c_1$ and $c_2$.
\end{restatable*}

\medskip
\textbf{Algorithmic Results:}
The above results imply that all results of~\citet{PLBNetworks} also hold \whp for Chung-Lu Random Graphs and Geometric Inhomogeneous Random Graphs and almost surely for Hyperbolic Random Graphs.
Therefore the problems transitive closure, maximum matching, determinant, PageRank, matrix inverse, counting triangles and maximum clique have faster algorithms on Chung-Lu and Geometric Inhomogeneous Random Graphs \whp and on Hyperbolic Random Graphs almost surely.

In this work we additionally consider the three classical $\NP$-complete problems \textsc{Minimum Dominating Set}(MDS), \textsc{Maximum Independent Set}(MIS) and \textsc{Minimum Vertex Cover}(MVC) on \PLBU networks. 
For the first two problems, positive results are already known for $(\alpha, \beta)$-Power Law Graphs, which
are a special case of graphs with the \PLBUL properties.Note that this deterministic graph class is much more restrictive and does \emph{not} cover typical real-world graphs.
On the contrary, our positive results only assume the \PLBU property. 
Our algorithmic results can therefore be applied to real-world networks after measuring the respective constants of the PLB-model. 
In section~\ref{sec:approx} we prove our main lemma, \lemref{main} (the potential volume lemma). 
Using the potential volume lemma, we prove lower bounds for MDS, MIS and MVC in the order of $\Theta(n)$ on \PLBU networks with exponent $\beta>2$. 
This essentially means, even taking all nodes as a solution gives a constant factor approximation. 
Furthermore, in \thmref{dsgreedy} we prove that the greedy algorithm actually achieves a better constant approximation ratio.
The positive results from \secref{approx} also hold for $(\alpha, \beta)$-Power Law Graphs.

\citet{PLBNetworks} proved that for \PLBUN networks with $\beta>3$ finding a maximum clique is solvable in polynomial time. 
This result gives rise to the question whether the \PLBN property can be helpful in solving other \NP-complete problems on power-law graphs in polynomial time.
In \secref{hardness} we consider the mentioned NP-Complete problems MDS, MIS and MVC and prove that these problems are \APX-hard even for \PLBULN networks with $\beta>2$.
Therefore, at least for the three problems we considered, even the \PLBN property is not enough to make those problems polynomial-time solvable.
As a side product we also get a lower-bound on the approximability of the respective problems under some complexity theoretical assumptions.
Since the negative results for $(\alpha, \beta)$-Power Law Graphs imply the same non-approximability on graphs with \PLBUL, we only consider graphs with \PLBULN in \secref{hardness}.

\medskip
\textbf{Dominating Set:}
Given a Graph $G=(V,E)$, a \textsc{Minimum Dominating Set} (MDS) is a subset $S\subseteq{V}$ of minimum size such that for each $v\in{V}$ either $v$ or a neighbor of $v$ is in $S$.
MDS cannot be approximated within a factor of $(1-\epsilon)\,\ln |V|$ for any $\epsilon>0$~\cite{Feige98} unless $\NP\subseteq\DTIME(|V|^{\log\log |V|})$ and not to within a factor of $\ln \Delta - c\ln\ln\Delta$ for some $c>0$~\cite{CC08} unless $\P=\NP$, although a simple greedy algorithm achieves an approximation ratio of $1+\ln \Delta$~\cite{kao2008encyclopedia}.
We also know that even for sparse graphs, MDS cannot be approximated within a factor of $o(\ln(n))$, 
since we could have a graph with a star of $n-\sqrt{n}$ nodes to which an arbitrary graph of the $\sqrt{n}$ remaining nodes is attached~\cite{lenzen2010minimum}.

MDS has already been studied in the context of $(\alpha, \beta)$-Power Law Graphs.
\citet{FPP08} showed that the problem remains $\NP$-hard for $\beta>0$.
\citet{SDY12} proved that there is no $\left(1+\tfrac{1}{3120\zeta(\beta)3^\beta}\right)$-approximation for $\beta>1$ unless $\P=\NP$. 
They also showed that the greedy algorithm achieves a constant approximation factor for $\beta>2$, showing that n this case the problem is \APX -hard .
\citet{GHK12} also proved a logarithmic lower bound on the approximation factor when $\beta\leq 2$.

For graphs with the \PLBU property we will show a lower bound on the size of the minimum dominating set in the range of $\Theta(n)$, which already gives us a constant factor approximation by taking all nodes.
In contrast to $(\alpha, \beta)$-Power Law Graphs the \PLBU property captures a wide range of real networks, making it possible to transfer our results to them.
All our upper bounds are in terms of the following two constants, which depend on the parameters $c_1$, $\beta$ and $t$ of the \PLBU property:
\[a_{\beta,t}:=\left(1+\frac{\beta-1}{\beta-2}\frac1{1-\left(\tfrac{t+2}{t+1}\right)^{1-\beta}}\right)\text{ and }b_{c_1,\beta,t}:=\left(c_1\tfrac{\beta-1}{\beta-2}\cdot 2^{\beta}\cdot(t+1)^{\beta-1}\right)^{\frac1{\beta-2}}.\]
In the rest of the paper we assume the parameters $c_1$, $\beta$ and $t$ to be constants. 
\begin{restatable*}{theorem}{stateminds}\label{thm:min-ds}  
	For a graph without loops and isolated vertices and with the \PLBU property with parameters $\beta>2$, $c_1>0$ and $t\ge0$, the minimum dominating set is of the size at least
	\[\left(2\cdot a_{\beta,t}\cdot b_{c_1,\beta,t}+1\right)^{-1}n = \Theta(n).\]
\end{restatable*}
Furthermore, we will show that the greedy algorithm actually achieves a lower constant approximation factor than the one we get from the above mentioned bound.
\begin{restatable*}{theorem}{statedsgreedy}\label{thm:dsgreedy}
  For a graph without loops and isolated vertices and with the \PLBU property with parameters $\beta>2$,  $c_1>0$ and $t\ge0$, the classical greedy algorithm for \textsc{Minimum Dominating Set} (cf.~\cite{DKH11}) has an approximation factor of at most
	\[\log_{3}(5)\cdot a_{\beta,t}\ln\left(b_{c_1,\beta,t}+1\right)+1=\Theta(1).\]
\end{restatable*}
Note that in networks with \PLBU the maximum degree can be $\Delta=\Theta(n^{\frac1{\beta-1}})$.
That means the simple bound for the greedy algorithm  gives us only an approximation factor of $\ln(\Delta+1)=\Theta(\log n)$.

\medskip
For the related problem \textsc{Minimum Connected Dominating Set} we prove the following constant approximation factor for
the greedy algorithm introduced by Ruan et al.~\cite{RDJ04}.
\begin{restatable*}{theorem}{statecds}\label{thm:Greedy_CDS}
  For a graph without loops and isolated vertices and with the \PLBU property with parameters $\beta>2$, $c_1>0$ and $t\ge0$, the greedy algorithm for \textsc{Minimum Connected Dominating Set} (cf.~\cite{RDJ04}) has an approximation factor of at most
	\[2+\ln\left(2 \cdot a_{\beta,t} \cdot b_{c_1,\beta,t}+1\right)=\Theta(1).\]
\end{restatable*}
Finally, we show that \textsc{Minimum Dominating Set} remains $\APX$-hard in networks with \PLBU and $\beta>2$, even with the \PLBL and \PLBN property.

\begin{table}[tb]
\begin{center}
\begin{tabular}{ll@{\ }ll@{\ }l}
\toprule
\textbf{Problem} & \multicolumn{2}{l}{\textbf{General Graph}} & \multicolumn{2}{l}{\textbf{Graphs with \textbf{\PLBU}}} \\
\midrule
 Minimum Dominating Set  & $\Oh(\ln\Delta)$ &~\cite{kao2008encyclopedia} & $\Theta_n(1)$ & [\thmref{dsgreedy}]\\ 
 Minimum Vertex Cover & $\Oh(\ln \Delta)$ & [Corollary \ref{greedVC}] & $\Theta_n(1)$ & [\corref{vc}] \\
 Maximum Independent Set & $\Oh(\Delta)$ &~\cite{DKH11} & $\Theta_n(1)$ & [\corref{MIS}] \\
 Minimum Connected Dominating Set & $\Oh(\ln \Delta)$ &~\cite{RDJ04} & $\Theta_n(1)$ & [\thmref{Greedy_CDS}] \\ 
   \bottomrule
\end{tabular}
 \caption{Comparison of the approximation ratios achieved by greedy algorithms on networks with an upper bound on the power-law degree distribution (\PLBU) and exponent $\beta>2$ and on general graphs. While on general graphs,
greedy achieves only a logarithmic or polynomial approximation, greedy achieves a constant-factor-approximation on graphs with \PLBU and $\beta>2$.}
 \label{Table1}
 \end{center}
 \end{table}

\medskip
\textbf{Independent Set:}
For a graph $G=(V,E)$, \textsc{Maximum Independent Set} (MIS) consists of finding a subset $S\subseteq{V}$ of maximum size, such that no two different vertices $u,v\in S$ are connected by an edge.
MIS cannot be approximated within a factor of $\Delta^\epsilon$ for some $\epsilon>0$ unless $\P=\NP$~\cite{AFWZ95}, although a simple greedy algorithm achieves an approximation factor of $\tfrac{\Delta+2}{3}$~\cite{HalldorssonR94}.
We also know from Tur{\'a}n's Theorem that every graph with an average degree of $\bar{d}$ has a maximum independent set of size at least $\tfrac{n}{\bar{d}+1}$.
This lower bound can already be achieved by the same greedy algorithm~\cite[Theorem~1]{HalldorssonR94}.

MIS has also been studied in the context of $(\alpha, \beta)$-Power Law Graphs.
\citet{FPP08} showed that the problem remains $\NP$-hard for $\beta>0$.
\citet{SDY12} proved that for $\beta>1$ there is no $\left(1+\tfrac{1}{1120\zeta(\beta)3^\beta}-\epsilon\right)$-approximation unless $\P=\NP$ and~\citet{HK15} gave the first non-constant bound on the approximation ratio of MIS for $\beta\leq 1$.

Since the \PLBU property with $\beta>2$ induces a constant average degree, the greedy algorithm already gives us a constant approximation factor for \textsc{Maximum Independent Set} on networks with these properties.
Although we can not give better bounds for the maximum independent set, \thmref{min-ds} immediately implies a lower bound for the size of \emph{all maximal independent sets}.
\begin{restatable*}{theorem}{stateislow}\label{thm:islow}
In a graph without loops and isolated vertices and with the \PLBU property with parameters $\beta>2$, $c_1>0$ and $t\ge0$, every maximal independent set is of size at least 
\[\left(2\cdot a_{\beta,t}\cdot b_{c_1,\beta,t}+1\right)^{-1}n=\Theta(n).\]
\end{restatable*}
It is easy to see that these lower bounds do not hold in sparse graphs in general,
since in a star the center node also constitutes a maximal independent set.

Furthermore, we show that \textsc{Maximum Independent Set} remains $\APX$-hard in networks with \PLBU and $\beta>2$, even with the \PLBL and \PLBN property.

\begin{table}[t]
\begin{center}
{
\begin{tabular}{ll@{\ }ll@{\ }l}
\toprule
\textbf{Problem} & \multicolumn{2}{l}{\textbf{General Graph}} & \multicolumn{2}{l}{\textbf{\textbf{Graph with \PLBULN}}} \\
\midrule
 Minimum Dominating Set (MDS)  &
 $\Omega(\ln \Delta)$
&~\cite{CC08} & $1+\Omega(1)$~[Theorem \ref{thm:ds-apx-simple}] \\ 
 Minimum Vertex Cover (MVC) &
$\geq 1.3606$ &~\cite{DinurSafra05}
&$1+\Omega(1)$~[Theorem \ref{thm:vc-apx-simple}]  \\
  Maximum Independent Set (MIS) &
  $\Omega(\poly(\Delta))$
& ~\cite{AFWZ95}
  & $1+\Omega(1)$~[Theorem \ref{thm:is-apx-simple}] \\ [1ex] 
\bottomrule
\end{tabular}}
 \caption{Comparison of the approximation lower bounds for polynomial-time algorithms (assuming $\P\neq\NP$) on networks with an upper (\PLBU) and lower (\PLBL) bound on the power-law degree distribution and with PLB neighborhoods (\PLBN) with the approximation lower bounds on general graphs. Even with the additional properties of \PLBL and \PLBN the problems on graphs with \PLBU remain APX-hard, i.e. these problems cannot admit a PTAS. Better lower bounds for each problem are in respective theorem, $\Omega(1)$ hides the \PLBL parameters $\beta, t$ and constant $c_2$. }
\end{center}
 \end{table}

\medskip
\textbf{Vertex Cover:}
Given a graph $G=(V,E)$, \textsc{Minimum Vertex Cover} (MVC) consists of finding a subset $S\subseteq V$ of minimum size such that each edge $e\in E$ is incident to at least one node from $S$.
MVC cannot be approximated within a factor of $10\sqrt{5}-21\approx1.3606$ unless \P=\NP , whereas the simple algorithm which greedily constructs a maximal matching achieves an approximation ratio of $2$~\cite{CombOpt}.
Unfortunatly, the greedy algorithm based on node degrees only achieves an approximation factor of $\ln \Delta$ [Corollary \ref{greedVC}].

\textsc{Minimum Vertex Cover} has also been studied in the context of $(\alpha, \beta)$-Power Law Graphs.
\citet{SDY12} proved that there is no PTAS for $\beta>1$ under the Unique Games Conjecture.

We can show that in networks with \PLBU and without isolated vertices the minimum vertex cover has to have a size of at least $\Theta(n)$.
This follows immediately from \thmref{min-ds}, since in a graph without isolated nodes every vertex cover is also a dominating set:
\begin{restatable*}{theorem}{statevc}\label{thm:vc}
In a graph without loops and isolated vertices and with the \PLBU property with parameters $\beta>2$, $c_1>0$ and $t\ge0$, the minimum vertex cover is of size at least 
\[\left(2\cdot a_{\beta,t} \cdot b_{c_1,\beta,t}+1\right)^{-1}n.\]
\end{restatable*}
Also, we show that \textsc{Minimum Vertex Cover} remains $\APX$-hard in networks with \PLBU and $\beta>2$, even with the \PLBL and \PLBN property.

\section{Preliminaries and Notation}
\label{sec:prelim}

We mostly consider undirected multigraphs $G=(V,E)$ without loops, where $V$ denotes the set of vertices and $E$ the multiset of edges with $n=|V|$.
In the following we will refer to multigraphs as graphs and state explicitly if we talk about simple graphs.
Throughout the paper we use $\deg(v)$ to denote the degree of node $v$, $d_i$ for the set of nodes of degree $i$ and $d_{\geq i}$ for the set of nodes of degree greater than or equal to $i$.
Furthermore, we use $d_{min}$ and $\Delta$ to denote the minimum and maximum degree of the graph respectively. 
For a $S\subseteq V$, the volume of S, denoted by \Vol(S) is the sum of degrees of vertices in $S$, $\Vol(S)=\sum_{v\in S}\deg(v)$. 
We use $b_i$ to denote the set of nodes $v\in V$ with $\deg(v)\in[2^i,2^{i+1})$ and for $v\in V$ we let $N^{+}(v)$ denote the \emph{inclusive neighborhood} of $v$ in $G$. 
If not stated otherwise $\log$ denotes the logarithm of base $2$.

Now we give a formal definition of the \PLBU, \PLBL and \PLBN properties.
\begin{definition}[\PLBU~\cite{PLBNetworks}]
Let $G$ be an undirected $n$-vertex graph and let $c_1>0$ be a universal constant. We say that $G$ is power law upper-bounded (\PLBU) for some parameters $1<\beta=\mathcal{O}(1)$ and $t\ge0$ if for every integer $d\ge0$, the number of vertices $v$, such that $\deg(v)\in\left[2^d,2^{d+1}\right)$ is at most
\[c_1n(t+1)^{\beta-1}\sum_{i=2^d}^{2^{d+1}-1}{(i+t)^{-\beta}}.\]
\end{definition}

\begin{definition}[\PLBL]
Let $G$ be an undirected $n$-vertex graph and let $c_2>0$ be a universal constant. We say that $G$ is power law lower-bounded (\PLBL) for some parameters $1<\beta=\mathcal{O}(1)$ and $t\ge 0$ if for every integer $\left\lfloor \log d_{min}\right\rfloor\le d \le \left\lfloor \log \Delta\right\rfloor$, the number of vertices $v$, such that $\deg(v)\in\left[2^d,2^{d+1}\right)$ is at least
\[c_2n(t+1)^{\beta-1}\sum_{i=2^d}^{2^{d+1}-1}{(i+t)^{-\beta}}.\]
\end{definition}

Since the \PLBU property alone can capture a much broader class of networks, for example empty graphs and rings, this lower-bound is important to restrict networks to real power-law networks. 
In the definition of \PLBL $d_{min}$ and $\Delta$ are necessary because in real world power law networks there are no nodes of lower or higher degree, respectively.

\begin{definition}[\PLBN~\cite{PLBNetworks}]
Let $G$ be an undirected $n$-vertex graph with \PLBU for some parameters $1<\beta=\mathcal{O}(1)$ and $t\ge0$.
We say that $G$ has PLB neighborhoods (\PLBN) if for every vertex $v$ of degree $k$, the number of neighbors of $v$ of degree at least $k$ is at most 
$c_3\max\left(\log n, (t+1)^{\beta-2}k\sum_{i=k}^{n-1}{i(i+t)^{-\beta}}\right)$ for some universal constant $c_3>0$.
\end{definition}

\medskip \noindent 

\section{Power-Law Random Graphs and the PLB properties}\label{sec:random}
In this section  we consider some well known power law random graph models and prove that \whp or \textit{almost surely} graphs generated by these models have \PLBU and \PLBN properties.
We chose Chung-Lu Random Graphs, Geometric Inhomogeneous Random Graphs, and Hyperbolic Random Graphs, because they are common models and rather easy to analyze.
Furthermore, they assume independence or some geometrically implied sparseness of edges, which is important for establishing the \PLBN property.

We need the following lemma, which is a more precise version of Lemma~2.2 from~\cite{PLBNetworks}. 
\begin{lemma}\label{lem:potenceUB}
Let $1\le a\le b/2$, for $a,b\in\N$, and let $c>0$ be a constant. Then 
\[a^{-c}\le\tfrac{c}{1-2^{-c}}\sum_{i=a}^{b-1}{i^{-c-1}}.\]
\end{lemma}
\begin{proof}
\[\sum_{i=a}^{b-1}{i^{-c-1}}\ge\int_{a}^{b} \! x^{-c-1} \, \mathrm{d}x=\tfrac1c\left(a^{-c}-b^{-c}\right)\ge\tfrac{1-2^{-c}}{c}\cdot a^{-c}.\qedhere\]
\end{proof}

\subsection*{$(\alpha, \beta)$-Power Law Graph}

\begin{definition}[$(\alpha,\beta)$-Power Law Graph~\cite{ACL01}]
An $(\alpha,\beta)$-Power Law Graph is an undirected multigraph with the following degree distribution depending on two given values $\alpha$ and $\beta$. For $1\le i\le\Delta=\left\lfloor e^{\alpha/\beta}\right\rfloor$ there are
$y_i=\left\lfloor \tfrac{e^\alpha}{i^\beta}\right\rfloor$ nodes of degree $i$.
\end{definition}
\begin{theorem}
\label{thmerab}
The $(\alpha,\beta)$-Power Law Graph with $\beta>1$ has the \PLBU property with $c_1=\tfrac{1}{\zeta(\beta)}$, $t=0$ and exponent $\beta$.
\end{theorem}

\begin{proof}
It holds that the number of nodes of degree between $2^d$ and $2^{d+1}-1$ is at most
\[e^\alpha\sum_{i=2^d}^{2^{d+1}-1}{i^{-\beta}} = \tfrac{n}{\zeta(\beta)}\sum_{i=2^d}^{2^{d+1}-1}{i^{-\beta}}\]
due to the definition of the degree distribution and the fact that $n=\zeta(\beta)e^\alpha$ for $\beta>1$.
\end{proof}

\begin{theorem}
The $(\alpha,\beta)$-Power Law Graph with $\beta>1$ has the \PLBL property with $c_1=\tfrac{1}{2\zeta(\beta)}$, $t=0$ and exponent $\beta$.
\end{theorem}

\begin{proof}
The number of nodes of degree $i$ is exactly $\left\lfloor \tfrac{e^\alpha}{i^\beta}\right\rfloor$.
Since $i\le\left\lfloor e^{\alpha/\beta}\right\rfloor$, this number is at least one.
Therefore $\left\lfloor \tfrac{e^\alpha}{i^\beta}\right\rfloor\ge\tfrac12\tfrac{e^\alpha}{i^\beta}$.
It now holds that the number of nodes of degree between $2^d$ and $2^{d+1}-1$ is at least
\[\tfrac{e^\alpha}2\sum_{i=2^d}^{2^{d+1}-1}{i^{-\beta}} = \tfrac{n}{2\zeta(\beta)}\sum_{i=2^d}^{2^{d+1}-1}{i^{-\beta}}\]
due to the definition of the degree distribution and the fact that $n=\zeta(\beta)e^\alpha$ for $\beta>1$.
\end{proof}

\begin{corollary}
A random $(\alpha,\beta)$-Power Law Graph with $\beta>1$ created with the Erased Configuration Model has the \PLBU and \PLBN properties with high probability.
\end{corollary}
\begin{proof}
\citet{PLBNetworks} proved that random networks created by the Erased Configuration Model whose prescribed degree sequence follows \PLBU, also follow \PLBU and \PLBN with high probability.
\end{proof}

\subsection*{Geometric Inhomogeneous Random Graphs}

\begin{definition}[Geometric Inhomogeneous Random Graphs (GIRGs)~\cite{BKL15}]
\label{girg}
A Geometric Inhomogeneous Random Graph is a simple graph $G=(V,E)$ with the following properties.
For $|V|=n$ let $w=(w_1,\cdots, w_n)$ be a sequence of positive weights. 
Let $W=\sum_{i=1}^n w_i $ the total weight. 
For any vertex $v$, draw a point $x_v \in \mathbb{T}^d $ uniformly and independently at random. 
We connect vertices $u\neq v$ independently with probability $p_{uv}=p_{uv}(r)$, which depends
on the weights $w_u$, $w_v$ and on the positions $x_u$, $x_v$, more precisely, on the distance $r=\left\|x_u-x_v\right\|$.
 We require for some constant $\alpha>1$ the following edge probability condition
\[p_{uv}=\Theta\Big(\ min\Big\{\frac{1}{||x_u-x_v||^{\alpha d}}\Big(\frac{w_u w_v}{W} \Big)^{\alpha},1\Big\} \Big).\]  
\end{definition}

\begin{definition}[General Power-law~\cite{BKL15}]\label{def:generalpl}
A weight sequence $\vec{w}$ is said to follow a general power-law with exponent $\beta > 2$ if $w_{\min}:=\min\left\{w_v\mid v\in V\right\}=\Omega(1)$ and if there is a $\bar{w}=\bar{w}(n)\ge n^{\omega(1/\log \log n)}$ such that for all constants $\eta>0$ there are $c_1,c_2>0$ with
\[c_1\tfrac{n}{w^{\beta-1+\eta}}\le \left|\left\{v\in V\mid w_v\ge w\right\}\right|\le c_2\tfrac{n}{w^{\beta-1-\eta}},\]
where the first inequality holds for all $w_{\min}\le w \le \bar{w}$ and the second holds for all $w\ge w_{\min}$.
\end{definition}

Now we are going to prove that GIRGs fulfill \PLBU and \PLBN. For this we need the following theorem and auxiliary lemmas by~\citet{BKL16}.
\begin{theorem}[\cite{BKL16}]
\label{thm:girgpl}
Let $G$ be a GIRG with a weight sequence that follows a general power-law with exponent $\beta$ and average degree $\Theta(1)$.
Then, with high probability the degree sequence of $G$ follows a power law with exponent $\beta$ and average degree $\Theta(1)$, i.e there exist constants $c_3,c_4 >0$ such that \whp
\[c_3\frac{n}{k^{\beta-1+\eta}}\leq |\{v\in V|deg(v)\geq k\}|\leq c_4 \frac{n}{k^{\beta-1-\eta}},\]
where the first inequality holds for all $1\le d\le\bar{w}$ and the second holds for all $d\ge1$.
\end{theorem}

The following three lemmas are necessary to prove \thmref{girg}. 

\begin{lemma}[\cite{BKL16}]
\label{lem:girg-indep}
Fix $u\in[n]$ and $x_u\in \mathcal{X}$. All edges $\left\{u,\ v\right\}$, $u\neq v$, are independently present with probability
\[\Pr\left[u\sim v\mid x_u\right]=\Theta(\Pr\left[u\sim v\right])=\Theta\big(\min\left\{1,\tfrac{w_u w_v}{W}\right\}\big).\]
\end{lemma}

\begin{lemma}[\cite{BKL16}]
\label{lem:girg-expec}
For any $v\in [n]$ in Geometric Inhomogeneous Random Graph,
\[\mathbb{E}[deg(v)]=\Theta(w_v).\]
\end{lemma}

The former two lemmas imply that we can use standard Chernoff bounds to bound node degrees, but we also need the following auxiliary lemma.

\begin{lemma}[\cite{BKL16}]
\label{lem:girg-total-weight}
Let $\vec{w}$ a general power-law weight sequence with exponent $\beta$.
Then the total weight satisfies $W=\Theta(n)$. Moreover, for all sufficiently small $\eta>0$,
\begin{enumerate}[(i)]
\item $W_{\ge w}=\Oh(nw^{2-\beta+\eta})$ for all $w\ge w_{\min}$,
\item $W_{\ge w}=\Omega(nw^{2-\beta-\eta})$ for all $w_{\min}\le w \le \bar{w}$,
\item $W_{\le w}=\Oh(n)$ for all $w$, and
\item $W_{\le w}=\Omega(n)$ for all $w=\omega(1)$.
\end{enumerate}
\end{lemma}

\stategirg
\begin{proof}
First, we show that $G$ fulfills \PLBU with high probability.
Let $k=2^d$. 
It now holds that
\begin{align}
|\{v\in V|deg(v)\geq k\}|&\leq c_4 \frac{n}{k^{\beta'-1-\eta}}\notag\\
&\leq c_4 n \tfrac{\beta-1-\eta}{1-2^{-\beta'+1+\eta}}\sum_{i=k}^{2k-1} i^{-\beta'+\eta}\notag
\end{align}
due to \thmref{girgpl} and \lemref{potenceUB}.
This means that $G$ has the \PLBU property with $\beta=\beta'-\eta$, $t=0$ and $c_1=c_4\tfrac{\beta-1-\eta}{1-2^{-\beta'+1+\eta}}$.

Now we show that $G$ also fulfills \PLBN with high probability.
We first bound the range into which the degree $\deg(v)$ of a node $v$ with weight $w_v$ can fall with high probability.
We pessimistically assume that $\deg(v)$ takes its lower bound, because then the number of possible neighbors of degree at least $\deg(v)$ is maximized.
Further, we assume that all other nodes' degrees take their respective upper bounds.
This gives a number of potential nodes with high enough weight.
Finally, we bound the number of these potential nodes that are neighbors of $v$.

Due to \lemref{girg-indep} we can use standard Chernoff bounds (cf.~\cite[Theorem~1.1]{DP09}) to bound the degrees of nodes. 
According to \lemref{girg-expec} there are constants $c_7, c_8 > 0$ such that
\[c_7\cdot w_v \le \Ex[\deg(v)]\le c_8\cdot w_v\]
holds for all $v\in V$.
Let $c$ an appropriately chosen constant.
For a node $v\in V$ with $w_v\ge c\ln n$ it holds that 
\[\Pr\left(\deg(v)>\tfrac32 \Ex\left[\deg(v)\right]\right)\le e^{-\frac{\Ex\left[\deg(v)\right]}{12}}\le e^{-\frac{c_8\cdot w_v}{12}} \le n^{-\tfrac{c\cdot c_8}{12}}\]
and
\[\Pr\left(\deg(v)<\tfrac12 \Ex\left[\deg(v)\right]\right)\le e^{-\frac{\Ex\left[\deg(v)\right]}{8}}\le e^{-\frac{c_8\cdot w_v}{8}} \le n^{-\tfrac{c\cdot c_8}{8}}.\]
For a sufficiently large constant $c$ it holds \whp that
\begin{equation}
\tfrac12\cdot c_7\cdot w_v\le\tfrac12 \Ex\left[\deg(v)\right]\le\deg(v)\le\tfrac32 \Ex\left[\deg(v)\right]\le \tfrac32\cdot c_8\cdot w_v.\label{eq:potNeigh}
\end{equation}

For nodes $v\in V$ with $w_v < c\ln n$ it holds that 
\[\Pr\left(\deg(v)>2e\cdot c_8\cdot c\ln n\right)\le 2^{-2e\cdot c_8\cdot c\ln n} = n^{-\tfrac{2e\cdot c\cdot c_8}{\log_{2}(e)}},\]
since $2e\cdot c_8\cdot c\ln n>2e\cdot c_8\cdot w_v\ge 2e\cdot\Ex\left[\deg(v)\right]$. 
For a sufficiently large constant $c$ it holds \whp that the degrees of these nodes are at most $2e\cdot c_8\cdot c\ln n=\mathcal{O}(\log n)$.
This already complies with the bound from \PLBN.

For nodes $v\in V$ with $c\ln n\le w_v < 4e\tfrac{c_8}{c_7}c\ln n$ it holds that $\deg(v)\le 6e\tfrac{{c_8}^2}{c_7}c\ln n=\mathcal{O}(\log n)$ \whp due to \ineq{potNeigh}.
This also complies with the bounds from \PLBN.

Now let us fix some $v\in V$ with $w_v \ge 4e\tfrac{c_8}{c_7}c\ln n$.
We can assume $\deg(v)\ge \tfrac{1}{2}c_7\cdot w_v\ge 2e\cdot c_8\cdot c\ln n$.
Due to this fact, no node $u$ with $w_u<c\ln n$ can reach a degree of $\deg(v)$ with high probability.
That means, the only nodes that can reach a degree of at least $\deg(v)$ \whp are those with $w_u\ge \tfrac13\tfrac{c_7}{{c_8}}w_v=:\hat{w}$ due to \ineq{potNeigh}.
These are the potential neighbors of $v$ with degree at least $\deg(v)$.
Let $X$ the number of edges between $v$ and these potential neighbors.
Now it holds that
\[\Ex\left[X\right]=\Theta\left(\tfrac{w_v}W\cdot W_{\ge\hat{w}}\right)\le\Oh(w_v\cdot \hat{w}^{2-\beta'+\eta})=\Oh(w_v^{3-\beta'+\eta})\]
due to \lemref{girg-indep} and \lemref{girg-total-weight}.
We can assume that the expected value is at most $c_9w_v^{3-\beta'+\eta}$.
Again, we can use Chernoff bounds to bound the number of these edges.
If $c_9\cdot w_v^{3-\beta'+\eta}< c\ln n$ it holds that
\[\Pr\left(X>2\cdot e\cdot c\ln n\right)\le 2^{2\cdot e\cdot c\ln n}.\]
If $c_9\cdot w_v^{3-\beta'+\eta}\ge c\ln n$ it holds that
\[\Pr\left(X>\tfrac32 c_9\cdot w_v^{3-\beta'+\eta}\right)\le e^{-\frac{c_9\cdot w_v^{3-\beta'+\eta}}{12}} \le n^{-\tfrac{c\cdot c_9}{12}}.\]
It now holds that $X=\Oh(\deg(v)^{3-\beta'+\eta}+\ln(n))$ \whp, since \whp $\deg(v)=\Theta(w_v)$.
Due to the requirement $\beta'-\eta>2$, it holds that
\[\sum_{i=k}^{n-1}{i^{1-\beta'+\eta}}\ge\int_{i=k}^{n}{i^{1-\beta'+\eta}}\ \mathrm{d}i=\tfrac1{\beta'-\eta-2}k^{2-\beta'+\eta}\left(1-\left(\tfrac kn\right)^{\beta'-\eta-2}\right).\]
Since in our case $k=\deg(v)\le\max_v\left\{\deg(v)\right\}=\Delta$ and $\Delta=\mathcal{O}(n^{\frac{1}{\beta'-\eta-1}})$ \whp, we get
\[\deg(v)^{2-\beta'+\eta}=\mathcal{O}\left(\sum_{i=\deg(v)}^{n-1}{i^{1-\beta'+\eta}}\right).\]
This implies that the number of neighbors of $v$ with degree at least $\deg(v)$ is at most $\mathcal{O}\left(\deg(v)\sum_{i=\deg(v)}^{n-1}{i^{1-\beta'+\eta}}+\ln(n)\right)$, which is at most
\[c_3\max\Big(\log(n), (t+1)^{\beta'-\eta-1}\deg(v)\sum_{i=\deg(v)}^{n-1}{i(i+t)^{-\beta'+\eta}}\Big)\]
 for $\beta=\beta'-\eta$, a suitable constant $c_3$ and $t=0$ as desired.
\end{proof}

\subsection*{Hyperbolic Random Graphs}
\begin{definition}(Hyperbolic Random Graph~\cite{krioukov2010hyperbolic})
Let $\alpha_{H}>0,$ $C_{H}\in \mathbb{R},$ $T_H>0$, $n\in \mathbb{N}$ and $R=2\log n + C_H$. The Hyperbolic Random Graph $G_{\alpha_H,C_H,T_H}(n)$ is a simple graph with vertex set V=[n] and the following properties:
\begin{itemize}
\item Every vertex $v\in [n]$ draws random coordinates independently at random $(r_v,\phi_v)$, where the angle $\pi_v$ is chosen uniformly at random in $[0,2\pi)$ and the radius $r_v\in [0,R]$ is random according to density $f(r)=\frac{\alpha_H \sinh(\alpha_H r)}{\cosh(\alpha_H R)-1}$. 
\item Every potential edge $e=\left\{u,v\right\}\in \binom{[n]}{2}$ is present independently with probability 
\[p_H(d(u,v))=(1+e^{\frac{d(u,v)-R}{2T_H}})^{-1}\]  
\end{itemize}
\end{definition}

\begin{lemma}[\cite{BKL15}]
\label{lmhrg}
Hyperbolic random graphs are a special case of GIRGs.
\end{lemma}
This lemma directly leads to the following consequence.

\statehyperbolic

\subsection*{Chung-Lu Random Graphs}
\begin{definition}[Chung-Lu Random Graph~\cite{chung2002average}]
A Chung-Lu Random Graph is a simple graph $G=(V,E)$.
Given a weight sequence $\mathbf{w}=(w_1,w_2,\ldots,w_n)$ the edges between nodes $v_i$ and $v_j$ exist independently with probability $p_{ij}$ proportional to $\min\left(1,\tfrac{w_iw_j}{W}\right)$, where $W=\sum_{i=1}^{n}w_i$.
\end{definition}
It has to be noted that the same proofs as for \thmref{girgpl} can be used for Chung-Lu random graphs, since all the necessary lemmas also hold for Chung-Lu Random Graphs (cf.~\cite{BKL16}).

\statechunglu

\section{Greedy Algorithms} \label{sec:approx}
In this section we try to understand why simple greedy algorithms work efficiently in practice. 
As in~\cite{DKH11} the basic idea of greedy algorithms can be summarized as follows:
\begin{itemize}
\item We define a potential function $f(S)$ on solution sets $S$.
\item Starting with $S=\emptyset$, we grow the solution set $S$ by adding to it, at each stage, an element that maximizes (or, minimizes) the value of $f(S\cup \{x\})$, until $f(S)$ reaches the maximum (or, respectively, minimum) value.   
\end{itemize}
\begin{definition}\label{defapp}
A greedy algorithm is an \emph{$\alpha$-approximation} for problem $P$ if it produces a solution set $S$ with 
$\alpha\geq \frac{|S|}{|\opt|}$ if P is a minimization problem and with $\alpha\geq \frac{|\opt|}{|S|}$ if P is a maximization problem.
\end{definition}
\subsection*{Analysis of Greedy Algorithms on \PLBU Networks}

This section will be dedicated to proving our Main~\lemref{main}.
From it we will be able to derive bounds on the size of solutions of covering problems as well as better approximation guarantees for the greedy dominating set algorithm.
\begin{restatable}[Potential Volume Lemma]{lemma}{statePVL}
\label{lem:main}
Let $G$ be a graph without loops and with the \PLBU property for some $\beta>2$, some constant $c_1>0$ and some constant $t\ge 0$.
Let $S$ be a solution set for which we can define a function $g\colon \R^+\to\R$ as continuously differentiable and $h(x):=g(x)+C$ for some constant $C$ such that
\begin{enumerate}[(i)]
\item $g$ non-decreasing,
\item $g(2x)\le c\cdot g(x)$ for all $x\ge 2$ and some constant $c>0$,
\item $g'(x)\le \tfrac{g(x)}{x}$,
\end{enumerate}
then it holds that $\sum_{x\in S}{h(\deg(x))}$ is at most
\[\left(c\left(1+\frac{\beta-1}{\beta-2}\frac1{1-\left(\tfrac{t+2}{t+1}\right)^{1-\beta}}\right)g\left(\left(c_1\tfrac{\beta-1}{\beta-2}\tfrac{n}{M}\cdot 2^{\beta-1}\cdot(t+1)^{\beta-1}\right)^{\frac1{\beta-2}}\right)+C\right)\cdot|S|,\]
where $M(n)\ge1$  is chosen such that $\sum_{x\in S}{\deg(x)}\ge M$.
\end{restatable}
\begin{proof}
Without loss of generality assume the nodes of $G$ were ordered by increasing degree, i.e.\ $V(G)=\left\{v_1,v_2,\ldots,v_n\right\}$ with $\deg(v_1)\ge\deg(v_2)\ge\ldots\ge\deg(v_n)$. 
Let $n':=2^{\left\lfloor \log(n-1)\right\rfloor+1}-1$. This is the maximum degree of the bucket an $(n-1)$-degree node is in.
For $j\in\N$ let 
\[D(j):=c_1\cdot n (t+1)^{\beta-1}\sum_{i=2^j}^{n'}{(i+t)^{-\beta}},\]
i.e.\ the maximum number of nodes of degree at least $2^j$ that $G$ can have according to the PLB property, and let 
\[s(\ell):=\min\left\{j\in\N\mid D(j)\le \ell\right\}.\]
We can interpret $s(\ell)$ as the index $j$ of the smallest bucket, such that the total number of nodes in buckets $j$ to $\left\lfloor \log(n-1)\right\rfloor$ is at most $\ell$.

It now holds for all $\ell\in\R$ with $\ell\le |S|$ that
\begin{align*}
\frac{\sum_{x\in S}{h(\deg(x))}}{|S|}
	&\le\frac{\sum_{i=1}^{\left\lfloor \ell\right\rfloor}{h(\deg(v_i))}+(\ell-\left\lfloor \ell\right\rfloor)h(\deg(v_{\left\lceil \ell\right\rceil}))}{\ell}
\end{align*}
due to the fact that $g$, and therefore also $h$, is non-decreasing. 
To upper bound the numerator on the right-hand side we assume that in each bucket we have the maximum (fractional) number of nodes of maximum degree, leading to at most
\[\sum_{j=s(\ell)}^{\left\lfloor \log(n-1)\right\rfloor}{\Big(c_1\cdot n (t+1)^{\beta-1}\sum_{i=2^j}^{2^{j+1}-1}{(i+t)^{-\beta}}\Big)h(2^{j+1}-1)}+(\ell-D(s(\ell)))h(2^{s(\ell)}-1).\]
Now let $k:=s(|S|)$.
For $\ell=D(k)$ we now get
\[\frac{\sum_{x\in S}{h(\deg(x))}}{|S|}\le\frac{c_1\cdot n (t+1)^{\beta-1}\sum_{j=k}^{\left\lfloor \log(n-1)\right\rfloor}{h(2^{j+1}-1)}\sum_{i=2^j}^{2^{j+1}-1}{(i+t)^{-\beta}}}{D(k)}\]
or equivalently
\begin{align}
\sum_{x\in S}{h(\deg(x))}
	&\le\tfrac{c_1\cdot n (t+1)^{\beta-1}\sum_{j=k}^{\left\lfloor \log(n-1)\right\rfloor}{h(2^{j+1}-1)}\sum_{i=2^j}^{2^{j+1}-1}{(i+t)^{-\beta}}}{D(k)}|S|\notag\\
	&=\left(\tfrac{c_1\cdot n (t+1)^{\beta-1}\sum_{j=k}^{\left\lfloor \log(n-1)\right\rfloor}{g(2^{j+1}-1)}\sum_{i=2^j}^{2^{j+1}-1}{(i+t)^{-\beta}}}{D(k)}+C\right)|S|\label{eq:approx},
\end{align}
since $h(x)=g(x)+C$ for some constant $C$ and the numerator for that second term just sums up to $C\cdot D(k)$.
One more note of caution: W.l.o.g.\ we assume $|S|\ge 1$. 
In the case of equality, we need $s(1)\le\left\lfloor \log(n-1)\right\rfloor$.
Otherwise the last bucket already contains more nodes than our solution, i.e. we could not take any complete bucket. 
For the desired inequality to hold, we must assure
\[D(\left\lfloor \log(n-1)\right\rfloor)=c_1\cdot n (t+1)^{\beta-1}\sum_{i=2^{\left\lfloor \log(n-1)\right\rfloor}}^{n'}{(i+t)^{-\beta}}\le1.\]
This is the case if
\begin{align*}
& c_1\cdot n (t+1)^{\beta-1}\sum_{i=2^{\left\lfloor \log(n-1)\right\rfloor}}^{n'}{(i+t)^{-\beta}}\\
&\le c_1\cdot n (t+1)^{\beta-1}\left((2^{\left\lfloor \log(n-1)\right\rfloor}+t)^{-\beta}+\tfrac{1}{\beta-1}{(2^{\left\lfloor \log(n-1)\right\rfloor}+t)^{1-\beta}}\right)\\
&\le c_1\cdot n (t+1)^{\beta-1}\tfrac{\beta}{\beta-1}2^{\beta-1}n^{1-\beta},
\end{align*}
since $2^{\left\lfloor \log(n-1)\right\rfloor}\ge\tfrac{n}2$. Now we can see that this is at most $1$ for
\[n\ge\left( c_1\cdot(t+1)^{\beta-1}\tfrac{\beta}{\beta-1}2^{\beta-1}\right)^{\frac{1}{\beta-2}}.\]

We start estimating \eq{approx} by deriving an upper bound on the numerator.
Using properties (i) and (ii) we derive
\[g(2^{j+1}-1)\le g(2^{j+1})\le c\cdot g(2^j) \le c\cdot g(i+t).\]
Plugging this into the numerator gives
\begin{align}
	&c_1\cdot n (t+1)^{\beta-1}\sum_{j=k}^{\left\lfloor \log(n-1)\right\rfloor}{g(2^{j+1}-1)}\sum_{i=2^j}^{2^{j+1}-1}{(i+t)^{-\beta}}\notag\\
	&\le c\cdot c_1 n (t+1)^{\beta-1}\sum_{j=k}^{\left\lfloor \log (n-1)\right\rfloor}{\sum_{i=2^j}^{2^{j+1}-1}{(i+t)^{-\beta}\cdot g(i+t)}}\notag\\
	&= c\cdot c_1 n (t+1)^{\beta-1}\sum_{i= 2^{k}}^{2^{\left\lfloor \log (n-1)\right\rfloor+1}-1}{(i+t)^{-\beta}\cdot g(i+t)}\label{eq:potential2}.
\end{align}
It is easy to check that the function $g(x)\cdot x^{-\beta}$ is non-increasing by property (iii) and the fact that $\beta>2$. Now we estimate the sum in \eq{potential2} by an integral
\begin{align}
\sum_{i=2^k}^{n'}{(i+t)^{-\beta}\cdot g(i+t)}
	& \le (2^k+t)^{-\beta}\cdot g(2^k+t)+\int_{x=2^k}^{n'}{(x+t)^{-\beta}\cdot g(x+t)}\ \mathrm{dx}\notag\\
	& = (2^k+t)^{-\beta}\cdot g(2^k+t)+\int_{x=2^k+t}^{n'+t}{x^{-\beta}\cdot g(x)}\ \mathrm{dx}.\label{eq:potential-sum}
\end{align}
Using integration by parts we get
\begin{align*}
\int_{x=2^k+t}^{n'+t}{(x)^{-\beta}\cdot g(x)}\ \mathrm{dx} 
	& = \tfrac{1}{1-\beta}\left[x^{1-\beta}\cdot g(x)\right]_{2^k+t}^{n'+t}-\tfrac{1}{1-\beta}\int_{2^k+t}^{n'+t}{x^{1-\beta}\cdot g'(x)}\ \mathrm{dx}\\
	& \le \tfrac{1}{\beta-1}(2^k+t)^{1-\beta}g(2^k+t)+\tfrac{1}{\beta-1}\int_{2^k+t}^{n'+t}{x^{1-\beta}g'(x)}\ \mathrm{dx},
\end{align*}
since $\beta>2$.
Due to property (iii) it holds that $x^{1-\beta}\cdot g'(x)\le x^{-\beta}\cdot  g(x)$, giving us
\[\int_{2^k+t}^{n'+t}{x^{1-\beta}\cdot g'(x)}\ \mathrm{dx}\le\int_{2^k+t}^{n'+t}{x^{-\beta}\cdot g(x)}\ \mathrm{dx}\]
and therefore
\[\int_{x=2^k+t}^{n'+t}{(x)^{-\beta}\cdot g(x)}\ \mathrm{dx} \le \tfrac{1}{\beta-2}(2^k+t)^{1-\beta}\cdot g(2^k+t).\]
Plugging this into \eq{potential-sum} yields
\begin{align*}
\sum_{i=2^k}^{n'}{(i+t)^{-\beta}\cdot g(i+t)}
	& \le (2^k+t)^{-\beta}\cdot g(2^k+t)+\tfrac{1}{\beta-2}(2^k+t)^{1-\beta}\cdot g(2^k+t).
	\end{align*}
Plugging this into \eq{potential2} now gives
\begin{align}
	&c_1\cdot n (t+1)^{\beta-1}\sum_{j=k}^{\left\lfloor \log(n-1)\right\rfloor}{g(2^{j+1}-1)}\sum_{i=2^j}^{2^{j+1}-1}{(i+t)^{-\beta}}\notag\\
	&\le c\cdot c_1 n (t+1)^{\beta-1}\cdot g(2^k+t)\left((2^k+t)^{-\beta}+\frac1{\beta-2}\cdot (2^k+t)^{1-\beta}\right).\label{eq:potential3}
\end{align}

Now we still need a lower bound on $D(k)$.
It holds that
\begin{align}
D(k)
	&= c_1\cdot n (t+1)^{\beta-1}\sum_{i=2^k}^{n'}{(i+t)^{-\beta}}\notag\\
	&\ge c_1\cdot n (t+1)^{\beta-1} \tfrac{1-\left(\tfrac{t+2}{t+1}\right)^{1-\beta}}{\beta-1}(2^k+t)^{1-\beta}\label{eq:degK},
\end{align}
where the last line follows by observing
\begin{align*}
\sum_{i=2^k}^{n'}{(i+t)^{-\beta}}
	& \ge \sum_{i=2^k}^{2^{k+1}-1}{(i+t)^{-\beta}}\\
	& \ge \int_{2^k}^{2^{k+1}}{(i+t)^{-\beta}}\\
	& =\tfrac{1}{\beta-1}\left((2^k+t)^{1-\beta}-(2^{k+1}+t)^{1-\beta}\right)\\
	& \ge \tfrac{1}{\beta-1}\left(1-\left(\tfrac{t+2}{t+1}\right)^{1-\beta}\right)(2^k+t)^{1-\beta}.
\end{align*}
This holds because $\tfrac{2^{k+1}+t}{2^k+t}\ge\tfrac{2+t}{1+t}$, since $2^k\ge 1$.

Plugging \eq{potential3} and \eq{degK} into \eq{approx} gives us an upper bound of
\begin{eqnarray}
\sum_{x\in S}{h(\deg(x))} & \le &\left(\frac{c\cdot c_1 n (t+1)^{\beta-1}\cdot g(2^k+t)\left((2^k+t)^{-\beta}+\frac{(2^k+t)^{1-\beta}}{\beta-2}\right)}{c_1\cdot n (t+1)^{\beta-1}\sum_{i=2^k}^{n'}{(i+t)^{-\beta}}}+C\right)\cdot|S|\notag\\
& \le & \left(c\left(1+\frac{\beta-1}{\beta-2}\frac1{1-\left(\tfrac{t+2}{t+1}\right)^{1-\beta}}\right)\cdot g(2^k+t)+C\right)\cdot|S|\label{eq:potential4}
\end{eqnarray}
It now suffices to find an upper bound for $2^k+t$, since $g(2^k+t)$ is non-decreasing.
Due to (iv) and the choice of $k$ it holds that
\begin{equation}
c_1\cdot n (t+1)^{\beta-1}\sum_{j=k-1}^{\left\lfloor \log(n-1)\right\rfloor}{(2^{j+1}-1)}\sum_{i=2^j}^{2^{j+1}-1}{(i+t)^{-\beta}}\ge \sum_{x\in S}{\deg(x)} \ge M.\label{eq:KBound}
\end{equation}

To upper bound the left-hand side, we can use \eq{potential3} with $g(x)=x$ and $2^{k-1}$ in place of $2^k$. 
It is easy to check, that this function satisfies (i), (ii) with $c=2$ and (iii) as needed.
This yields
\begin{align*}
M
	&\le2\cdot c_1\tfrac{\beta-1}{\beta-2} n (t+1)^{\beta-1}\cdot (2^{k-1}+t)^{1-\beta}\cdot (2^{k-1}+t)\\
	&=2\cdot c_1\tfrac{\beta-1}{\beta-2} n (t+1)^{\beta-1}\cdot (2^{k-1}+t)^{2-\beta}\\
	&\le2\cdot c_1\tfrac{\beta-1}{\beta-2} n (t+1)^{\beta-1}\cdot 2^{\beta-2}\cdot(2^{k}+t)^{2-\beta}
\end{align*}
or equivalently
\begin{equation}
(2^{k}+t)\le \left(c_1\tfrac{\beta-1}{\beta-2}\tfrac{n}{M}\cdot 2^{\beta-1}\cdot(t+1)^{\beta-1}\right)^{\frac1{\beta-2}}.\label{eq:KTBound}
\end{equation}

Now we can plug \eq{KTBound} into \eq{potential4} to get the result as desired.
\end{proof}

\subsection{Minimum Dominating Set}

The idea for lower-bounding the size of a dominating set is essentially the same as the one by~\citet{SDY12} and by~\citet{GHK12} in the context of $(\alpha,\beta)$-Power-Law Graphs.
Finally, we will show that 

\stateminds
\begin{proof}
	Let $\opt$ denote an arbitrary minimum dominating set.
	It holds that
	\[\sum_{x\in \opt}{\deg(x)+1}\ge n \]
	and since we assume that there are no nodes of degree $0$, it also holds that
	\[\sum_{x\in \opt}{\deg(x)}\ge \frac{n}{2},\]
	giving us (iv) with $M:=\tfrac{n}{2}$.
	We can choose $h(x):=x+1$ with $g(x)=x$.
	Now $g$ satisfies (i), (ii) with $c=2$ and (iii).
	With \lemref{main} we can now derive
	\begin{align*}
	n
	& \le \sum_{x\in \opt}{\deg(x)+1} = \sum_{x\in \opt}{h(\deg(x))}\\
	& \le \left(2\left(1+\frac{\beta-1}{\beta-2}\frac1{1-\left(\tfrac{t+2}{t+1}\right)^{1-\beta}}\right) \left(\left(c_1\tfrac{\beta-1}{\beta-2}\cdot 2^{\beta}\cdot(t+1)^{\beta-1}\right)^{\frac1{\beta-2}}\right)+1\right)\cdot |\opt|\qedhere
		\end{align*}
 \end{proof}

\begin{corollary}\label{cor:ds}
For a graph without loops and isolated vertices and with the \PLBU property with parameters $\beta>2$, $c_1>0$ and $t\ge0$, every dominating set has an approximation factor of at most
\[2\cdot a_{\beta,t} \cdot b_{c_1,\beta,t}+1\]
\end{corollary}

\corref{ds} says that simply taking all nodes already gives a constant approximation factor, but now we want to show that using the classical greedy algorithm actually guarantees an even better approximation factor.
To understand what happens, we shortly recap the algorithm.

\begin{algorithm}[t]
\caption{Greedy Dominating Set}
\label{alg:simple}
\begin{algorithmic}[1]
\Require undirected graph $G=(V,E)$
\State $C\gets\emptyset$
\State $D\gets\emptyset$
\While{$|D|<|V|$}
  \State $u\gets\argmax_{v\in \left(V\setminus C\right)}\left(N^{+}(v)\setminus D\right)$
  \State $C\gets C\cup\left\{u\right\}$
	\State $D\gets D\cup N^{+}(u)$
\EndWhile\\
\Return $C$
\end{algorithmic}
\end{algorithm}

The following inequality can be derived from an adaptation of the proof for the greedy \textsc{Set Cover} algorithm to the case of unweighted \textsc{Dominating Set}.
\begin{restatable}[\cite{kao2008encyclopedia}]{theorem}{stategreedyDS}\label{thm:greedyIneq}
Let $S$ the solution of the greedy algorithm and $\opt$ an optimal solution for \textsc{Dominating Set}.
Then it holds that
\[|C|\le\sum_{x\in \opt}{H_{\deg(x)+1}},\]
where $H_k$ is the $k$-th harmonic number.
\end{restatable}
\begin{proof}
The idea of the proof is to distribute the cost of taking a node $v\in C$ amongst the nodes that are newly dominated by $v$.
For example, if the algorithm chooses a node $v$ which newly dominates $v$, $v_1$, $v_2$ and $v_3$, the four nodes each get a cost of $1/4$.
At the end of the algorithm it holds that $\sum_{v\in V}{c(v)}=|C|$.

Now we look at the optimal solution $\opt$.
Since all nodes $v\in V$ have to be dominated by at least one $x\in \opt$, we can assign each node to exactly one $x\in \opt$ in its neighborhood, i.e. we partition the graph into stars $S(x)$ with the nodes $x$ of the optimal solution as their centers.
Now choose one $x\in \opt$ arbitrary but fixed.
Let us have a look at the time a node $u\in S(x)$ gets dominated.
Let $d(x)$ the number of non-dominated nodes from $S(x)$ right before $u$ gets dominated.
Due to the choice of the algorithm, a node $v$ had to be chosen which dominated at least $d(x)$ nodes.
This means $u$ gets a cost $c(u)$ of at most $1/w(x)$.
Now we look at the nodes from $S(x)$ in reverse order of them getting dominated in the algorithm.
The last node to get dominated has a cost of at most $1$, the next-to-last node gets a cost of at most $1/2$ and so on.
Since $|S(x)|\le\deg(x)+1$ the costs to cover $S(x)$ are at most $\tfrac{1}{\deg(x)+1}+\tfrac{1}{\deg(x)}+\ldots+1=H_{\deg(x)+1}$.
This gives us the inequality
\[|C|=\sum{v\in V}{c(v)}=\sum_{x\in \opt}{\sum_{u\in S(x)}{c(u)}}\sum_{x\in \opt}{H_{\deg(x)+1}}\]
as desired.
\end{proof}
From the former theorem, one can easily derive the following corollary.

\begin{corollary}
The greedy algorithm gives a $H_{\Delta+1}$-approximation for \textsc{Dominating Set}, where $\Delta$ is the maximum degree of the graph.
\end{corollary}

By using the inequality from \thmref{greedyIneq} together with the Potential Volume Lemma, we can derive the following approximation factor for the greedy algorithm.
\statedsgreedy
\begin{proof}
From the analysis of the greedy algorithm we know that for its solution $C$ and an optimal solution $\opt$ it holds that
\[|C|\le\sum_{x\in \opt}{H_{\deg(x)+1}}\le\sum_{x\in \opt}{\ln(\deg(x)+1)+1},\]
where $H_k$ denotes the $k$-th harmonic number.
We can now choose $h(x)=g(x)+1$ with $g(x)=\ln(x+1)$.
$g(x)$ satisfies (i), (ii) with $c=\log_{3}(5)$ and (iii).
As we assume there to be no nodes of degree $0$, it holds that
\[\sum_{x\in \opt}{\deg(x)}\ge\frac{n}{2}=:M,\]
since all nodes have to be covered.
We can now use \lemref{main} with $S=\opt$ to derive that
\[|C|\le \left(\log_{3}(5)\left(1+\frac{\beta-1}{\beta-2}\frac1{1-\left(\tfrac{t+2}{t+1}\right)^{1-\beta}}\right)\ln\left(\left(c_1\tfrac{\beta-1}{\beta-2}\cdot 2^{\beta}\cdot(t+1)^{\beta-1}\right)^{\frac1{\beta-2}}+1\right)+1\right)|\opt|.\qedhere\]
\end{proof}
For \textsc{Minimum Connected Dominating Set} we get a very similar bound.

\statecds
\begin{proof}
From~\cite[Theorem~3.4]{RDJ04} we know that for the solution $C$ of the greedy algorithm and an optimal solution $\opt$ it holds that
\[|C|\le\left(2+\ln\left(\frac{n}{|\opt|}\right)\right)|\opt|\leq \left(2+\ln\left(\frac{\sum_{x\in\opt}{\deg(x)+1}}{|\opt|}\right)\right)|\opt|.\]
We can now choose $h(x)=g(x)+1$ with $g(x)=x$.
$g(x)$ satisfies (i), (ii) with $c=2$ and (iii).
As we assume there to be no nodes of degree $0$, it holds that
\[\sum_{x\in \opt}{\deg(x)}\ge\frac{n}{2}=:M,\]
since all nodes have to be covered.
We can now use \lemref{main} with $S=\opt$ to derive that
\[|C|\le \left(2+\ln\left(2\left(1+\frac{\beta-1}{\beta-2}\frac1{1-\left(\tfrac{t+2}{t+1}\right)^{1-\beta}}\right) \left(\left(c_1\tfrac{\beta-1}{\beta-2}\cdot 2^{\beta}\cdot(t+1)^{\beta-1}\right)^{\frac1{\beta-2}}\right)+1\right)\right)|\opt|.\qedhere\]
\end{proof}

\subsection{Maximum Independent Set}

\begin{theorem}[\cite{Sak03}]
The greedy algorithm which prefers smallest node degrees gives a $(\Delta+1)$-approximation for MIS in graphs of degree at most $\Delta$.
\end{theorem}

For networks with the \PLBL property only, we can already derive the following lower bound on the optimal solution.
\begin{lemma}\label{lem:islow}
A graph with the \PLBL property with parameters $\beta>2$, $c_2>0$ and $t\ge0$, has an independent set of size at least $\frac{c_2(t+1)^{\beta-1}}{(t+d_{min})^{\beta}(d_{min}+1)}\cdot n$
or of size at least $\frac{c_2}{(t+1)}\cdot n$ if we assume $G$ to be connected and $d_{min}=1$.
\end{lemma}
\begin{proof}
This is easy to see by just counting the number of nodes of degree $d_{min}$.
There are at least
\[c_2 n (t+1)^{\beta-1}(t+d_{min})^{-\beta}=\frac{c_2(t+1)^{\beta-1}}{(t+d_{min})^{\beta}}n\]
of these nodes.
Since each of these nodes can have at most $d_{min}$ other nodes of the same degree as a neighbor, the independent set is at least of size $\frac{c_2(t+1)^{\beta-1}}{(t+d_{min})^{\beta}(d_{min}+1)}n$.
If $d_{min}=1$ and $G$ is connected, none of the degree-1 nodes can be neighbors, thus giving us $c_2 n (t+1)^{\beta-1}(t+d_{min})^{-\beta}=\frac{c_2}{(t+1)}n$.
\end{proof}
We can even go a step further and show that \emph{all} maximal independent sets have to be quite big, even if we only have the \PLBU property.
\stateislow
\begin{proof}
It holds that every maximal independent set $S$ is also a dominating set.
Due to \thmref{min-ds}, the size of the minimum dominating set is at least
\[\left(2\cdot a_{\beta,t} \cdot b_{c_1,\beta,t}+1\right)^{-1}n=\Theta(n),\]
giving us the result.
\end{proof}

\begin{corollary}\label{cor:MIS}
In a graph without loops and isolated vertices and with the \PLBU property with parameters $\beta>2$, $c_1>0$ and $t\ge0$, every maximal independent set has an approximation factor of at most
\[2\cdot a_{\beta,t} \cdot b_{c_1,\beta,t}+1.\]
\end{corollary}

\subsection{Vertex Cover}
First of all, it has to be noted, that one can use a greedy algorithm similar to \algref{simple} to achieve the following result.

\begin{corollary}\label{greedVC}
The greedy algorithm which prefers highest node degrees gives a \mbox{$H_{\Delta}$-approximation} for \textsc{Vertex Cover}, where $\Delta$ is the maximum degree of the graph.
\end{corollary}

This is clearly inferior to the simpler algorithm which achieves a $2$-approximation.
Nevertheless, from the results we know about \textsc{Dominating Set}, we can also derive some results about \textsc{Vertex Cover} in graphs without isolated vertices.
\statevc
\begin{proof}
The bound follows from \thmref{min-ds}, since every vertex cover in a graph without isolated vertices is a dominating set.
\end{proof}

\begin{corollary}\label{cor:vc}
For a graph without loops and isolated vertices and with the \PLBU property with parameters $\beta>2$, $c_1>0$ and $t\ge0$, every vertex cover has an approximation factor of at most
\[2\cdot a_{\beta,t} \cdot b_{c_1,\beta,t}+1.\]
\end{corollary}
\section{Hardness of Approximation}\label{sec:hardness}
In this section we show hardness of approximation once for  multigraphs and once for simple graphs with \PLBULN property.
The proofs for multigraphs use the embedding techniques of~\citet{SDY12}, while the proofs for the simple graphs employ our own embedding technique.
Let us start with a few definitions.
\begin{definition}[{\bf Class $\APX$}~\cite{DKH11}]  
$\APX$ is the class of all $\NPO$ problems that have polynomial-time $r$-approximation for some constant $r>1$.
\end{definition}
The notion of hardness for the approximation class $\APX$  is similar to the hardness of the class $\NP$, but instead of employing a polynomial time reduction it uses an approximation-preserving $\PTAS$-reduction. 
\begin{definition}[{\bf $\APX$-Completeness}~\cite{DKH11}]  
An optimization problem A is said to be $\APX$-complete iff A is in $\APX$ and every problem in the class $\APX$ is $\PTAS$-reducible to A.
\end{definition} 

\subsection{Approximation Hardness for Multigraphs}

In this section we adapt the cycle embedding technique used by~\citet{SDY12} to prove hardness of approximation for the optimization problems we consider.

\begin{definition}[Embedded-Approximation-Preserving Reduction~\cite{SDY12}]\label{def:embedding}
 Given an optimal substructure problem $O$, a reduction from an instance on graph $G =(V, E)$ to another instance on a
 (power law) graph $G' = (V', E')$ is called \emph{embedded approximation-preserving} if it satisfies the following properties:
\begin{enumerate}[(1)]
\item $G$ is a subset of maximal connected components of $G'$;
\item The optimal solution of $O$ on $G'$, \opt$(G')$, is upper bounded by $C\cdot $\opt$(G)$ where $C$ is a constant correspondent to the
growth of the optimal solution.
\end{enumerate}
\end{definition}

Having shown an embedded-approximation-preserving reduction, we can use the following lemma to show hardness of approximation.

\begin{lemma}[\cite{SDY12}]\label{lem:embedding}
Given an optimal substructure problem $O$, if there exists an embedded-approximation-preserving reduction from a graph $G$ to another graph $G'$, we can extract the inapproximability factor
$\delta$ of $O$ on $G'$ using $\varepsilon$-inapproximability of $O$ on $G$, where $\delta$ is lower bounded by $\frac{\varepsilon C}{(C-1)\varepsilon+1}$ when $O$ is a maximization problem and by $\frac{\varepsilon+C-1}{C}$ when $O$ is a minimization problem. \end{lemma}

We will use this framework as follows:
First, we show how to embed cubic graphs into graphs with \PLBU, \PLBL and \PLBN.
Then, we derive the value of $C$ as in \defref{embedding} for each problem we consider.
Last, we use \lemref{embedding} together with the known inapproximability results for the considered problems on cubic graphs to derive the approximation hardness on graphs with \PLBU, \PLBL and \PLBN.

We start by showing the embedding of cubic graphs into graphs with \PLBU, \PLBL and \PLBN.
In the embedding, we will use the following gadget to fill up the degree sequence of our graphs.

\begin{definition}[$\vec{d}$-Regular cycle $RC_n^{\vec{d}}$~\cite {SDY12}] Given a degree sequence $\vec{d}=(d_1,\cdots ,d_n)$, a \emph{$\vec{d}$-regular cycle $RC_n^{\vec{d}}$} is composed of two cycles. Each cycle has $n$ vertices and the two $i^{th}$ vertices in each cycle are adjacent to each other by $d_i-2$ multi-edges. That is, a $\vec{d}$-regular cycle has $2n$ vertices and the two $i^{th}$ vertices have degree $d_i$.  
\end{definition} 
\begin{restatable}{lemma}{statecubictoPLB}\label{lem:cubictoPLB}
Any cubic graph $G$ can be embedded into a graph $G_{PLB}$ having the \PLBU, \PLBL and \PLBN properties for any $\beta>1$ and any $t\ge 0$.
\end{restatable}
\begin{proof}
Suppose we are given $\beta$ and $t$. 
We now want to determine $c_1$ and $c_2$ of \PLBU and \PLBL respectively.
Let $n$ be the number of nodes in graph $G$ and let $N=cn$ be the number of nodes in $G_{PLB}$ for some constant $c$ to be determined.
Also, we have to ensure that $N-n$ is even to get a valid degree sequence since our gadgets always have an even number of nodes.
To hide a cubic graph in the respective bucket of $G_{PLB}$, we need
\begin{equation}\label{eq:embed-n}
c_1N(t+1)^{\beta-1}\sum_{i=2}^{3}(i+t)^{-\beta} = c_1N(t+1)^{\beta-1}\left(\frac{1}{(2+t)^\beta}+\frac{1}{(3+t)^\beta}\right) \ge n.
\end{equation}
Also, we have to ensure to choose $c_1$ big enough so that the bucket containing the maximum degree $\Delta$ can hold two vertices.
Otherwise we could not hide an appropriate $\vec{d}$-Regular cycle in that bucket, resulting in an empty bucket, which violates the \PLBL property for $c_2>0$.
This second condition implies
\begin{equation}\label{eq:embed-delta}
c_1N(t+1)^{\beta-1}\sum_{i=2^{\left\lfloor \log{\Delta}\right\rfloor}}^{2^{\left\lfloor \log{\Delta}\right\rfloor+1}-1}(i+t)^{-\beta} \ge 2.
\end{equation}
As we will see, we can choose the constant $c_1$ arbitrarily large, so the former conditions are no real restrictions.
Then we choose the maximum degree $\Delta$ such that
\[d_{max(G_{PLB})}=(cn)^{\frac{1}{\beta-1}}\]
and $d_{min}=1$.

In our embedding we first compute for each bucket $i$ the number of nodes necessary to reach the bucket's \PLBL bound.
Then we add node pairs to arbitrary buckets $d\ge 2$ as long as their \PLBU bounds are not violated and until we reach exactly $N$ nodes.
Then we connect these additional nodes.
Those of bucket $1$ are connected to form a cycle, while those for bucket $i$ with $i>1$ are connected to form a $\vec{d}$-Regular cycle with $\vec{d}=(2^i,\cdots, 2^i)$.
By filling a bucket (other than bucket $1$) we might deviate by at most two from the lower bound of that bucket, whereas bucket one gets at least $n$ nodes.
To ensure that we can add nodes until we have exactly $N$, we need the following inequality to hold true
\begin{align*}
& n+\sum_{d=0}^{\left\lfloor \log \Delta\right\rfloor}\left(2+c_2N(t+1)^{\beta-1}\sum_{i=2^d}^{2^{d+1}-1}(i+t)^{-\beta}\right)\\
& \le \frac{N}{c} + 2\log N^{\frac1{\beta-1}} + c_2N(t+1)^{-1} + \frac{c_2}{\beta-1}N\\
& \le N\left(\frac1c+\eta+\frac{c_2}{t+1}+\frac{c_2}{\beta-1}\right)\\
& \le N,
\end{align*}
i.e. after filling all buckets to their lower bound, there is still some slack until we reach $N$.
From this last condition we can derive
\[c\ge1+\frac{\eta+c_2\left(\frac{1}{t+1}+\frac{1}{\beta-1}\right)}{1-\eta-c_2\left(\frac{1}{t+1}+\frac{1}{\beta-1}\right)}>1+\frac{c_2\left(\frac{1}{t+1}+\frac{1}{\beta-1}\right)}{1-c_2\left(\frac{1}{t+1}+\frac{1}{\beta-1}\right)},\]
since $\eta$ can be arbitrarily small.
We choose $\eta=c_2\left(\frac{1}{t+1}+\frac{1}{\beta-1}\right)$ to obtain
\[c=1+\frac{2c_2\left(\frac{1}{t+1}+\frac{1}{\beta-1}\right)}{1-2c_2\left(\frac{1}{t+1}+\frac{1}{\beta-1}\right)}.\]

Now we can essentially choose $c_1$ arbitrarily large and $c_2$ arbitrarily small, guaranteeing a large enough gap to have a valid degree sequence and to guarantee $c>1$.
At the same time our choice of $c$ guarantees that we can fill the graph with exactly $N$ nodes.
Furthermore, since every node has a constant number of neighbors, $G_{PLB}$ also fulfills \PLBN, which always allows us at least $c_3 \log N$ many neighbors.
\end{proof}

\subsubsection{Dominating Set}

It has been shown by~\citet{SDY12} that the MDS can be found in polynomial time in any $\vec{d}$-Regular cycle.
The same holds true for cycles.
For the latter the size of an MDS is $\left\lceil \tfrac{n}3\right\rceil$.
For a $\vec{d}$-Regular cycle the following lemma gives the size of an MDS, the proof of which is an easy exercise and omitted for brevity.
\begin{lemma}\label{lem:d-reg-ds}
An $RC_n^{\vec{d}}$ has a minimum dominating set of size $\left\lceil \tfrac{|V(RC_n^{\vec{d}})|}4\right\rceil+1$ if $n=4i+2$ for some $i\in\N$ and of size $\left\lceil \tfrac{|V(RC_n^{\vec{d}})|}4\right\rceil$ otherwise.
\end{lemma}
To use the embedding framework as described, we also need the following inapproximability result.
\begin{theorem}[\cite{AK97, CC08}]\label{thm:cubic-MDS}
In $3$-bounded graphs it is $\NP$-hard to approximate MDS within a factor of $\frac{391}{390}$. 
\end{theorem}
Now we can state the desired hardness result:
\begin{restatable}{theorem}{stateapxds}\label{thm:ds-apx}
For every $\beta>1$ and every $t\ge0$ \textsc{Minimum Dominating Set} cannot be approximated to within a factor of $1+\left(130\cdot\left(4\frac{2c_2\left(\frac{1}{t+1}+\frac{1}{\beta-1}\right)}{1-2c_2\left(\frac{1}{t+1}+\frac{1}{\beta-1}\right)}+15\right)\right)^{-1}$ on graphs with \PLBU, \PLBL and \PLBN, unless $\P=\NP$.
\end{restatable}
\begin{proof} 
We give an $L$-reduction from a cubic graph $G$ to a graph $G_{PLB}$ with the \PLBU, \PLBL and \PLBN properties.
From \lemref{cubictoPLB} we know that there is an embedding of $G$ into $G_{PLB}$.
Let $\opt(G)$ and $\opt(G_{PLB})$ denote the size of a minimum dominating set for $G$ and $G_{PLB}$ respectively.
Let $b_i$ be the set of nodes in PLB bucket $i$, i.e. the set of nodes $v\in V$ with $\deg(v)\in\left[2^i,2^{i+1}-1\right]$.
We know that $\opt(G)\geq \frac{n}{4}$ and from \lemref{d-reg-ds} we know the size of the optimal solution for $G_{PLB}\text{\textbackslash}G$.
It holds that
\begin{align*}
\opt(G_{PLB})&=\opt(G)+\opt(G_{PLB}\text{\textbackslash}G)\\
&=\opt(G)+\left\lceil \frac{|b_1|-n}{3}\right\rceil+\sum_{i=2}^{\log \Delta}\left(\frac{|b_i|}{4}+1\right) \\
&\leq \opt(G)+\frac{|b_1|-n}{3}+1+\frac{N-|b_1|}{4}+\log N\\
&\leq \opt(G)+\frac{4N-4n}{12}+1+\log N\\
&\leq \opt(G)+\frac{n(c-1)}{3}+1+\log n+\log c\\
&\leq \opt(G)+\frac{n(c-1)}{3}+n\\
&\leq \opt(G)+\frac{4(c-1)\opt(G)}{3}+4\opt(G)\\
&=\frac{4c+11}{3}\opt(G),
\end{align*}
where we used our upper bounds on the size of optimal solutions in cycles and $\vec{d}$-Regular cycles in line~2.
That is, $C=\frac{4c+11}{3}$ in the context of \defref{embedding} and \lemref{embedding}.
Due to \thmref{cubic-MDS} it also holds that $\varepsilon=\frac{391}{390}$ in the context of \lemref{embedding}.
This gives us an approximation hardness of
\begin{align*}
1+\frac{\varepsilon-1}{C}
	&=1+\frac{3}{390\cdot(4c+11)}\\
	& =1+\frac{3}{390\cdot\left(4\frac{2c_2\left(\frac{1}{t+1}+\frac{1}{\beta-1}\right)}{1-2c_2\left(\frac{1}{t+1}+\frac{1}{\beta-1}\right)}+15\right)}\\
	& =1+\left(130\cdot\left(4\frac{2c_2\left(\frac{1}{t+1}+\frac{1}{\beta-1}\right)}{1-2c_2\left(\frac{1}{t+1}+\frac{1}{\beta-1}\right)}+15\right)\right)^{-1}
\end{align*}
due to our choice of $c$ in \lemref{cubictoPLB}.
\end{proof}

\subsubsection{Independent Set}
Again,~\citet{SDY12} showed that an MIS can be found in polynomial time in any $\vec{d}$-Regular cycle.
For cycles the same holds true.
The size of an MIS in a cycle is $\left\lfloor \tfrac{n}2\right\rfloor$.
The following lemma gives the size of an MIS in a $\vec{d}$-Regular cycle. 
As for MDS, we omit this simple proof for the sake of brevity.

\begin{lemma}\label{lem:d-reg-is}
An $RC_n^{\vec{d}}$ has a maximum independent set of size $\frac{|V(RC_n^{\vec{d}})|}{2}$ if $n$ is even and of size $\frac{|V(RC_n^{\vec{d}})|}{2}-1$ if $n$ is odd.
\end{lemma}

To use the framework of~\citet{SDY12} we need the following inapproximability result for MIS.
\begin{theorem}[\cite{AK97, BK99}]\label{thm:cubic-MIS}
In $3$-bounded graphs it is $\NP$-hard to approximate MIS within a factor of $\frac{140}{139}-\gamma$ for any $\gamma>0$.
\end{theorem}

We can now state the main result of this subsection.
\begin{restatable}{theorem}{stateapxis}\label{thm:is-apx}
For every $\beta>1$ and every $t\ge0$ \textsc{Maximum Independent Set} cannot be approximated to within a factor of $1+\frac{\left(\frac{1}{139}-\gamma\right)\left(1-2c_2\left(\frac{1}{t+1}+\frac{1}{\beta-1}\right)\right)}{4c_2\left(\frac{1}{t+1}+\frac{1}{\beta-1}\right)\left(\frac{140}{139}-\gamma\right)+1-2c_2\left(\frac{1}{t+1}+\frac{1}{\beta-1}\right)}$ for any $\gamma>0$ on graphs with \PLBU, \PLBL and \PLBN unless $\P=\NP$.
\end{restatable}
\begin{proof}
Again, we embed a cubic graph $G$ into a graph $G_{PLB}$ which fulfills \PLBU, \PLBL and \PLBN using \lemref{cubictoPLB}. 
Let $\opt(G)$ and $\opt(G_{PLB})$ denote the size of a maximum independent set of $G$ and $G_{PLB}$ respectively. 
We know that $\opt(G)\geq \frac{n}{4}$ and from \lemref{d-reg-is} we know the size of the optimal solution for $G_{PLB}\text{\textbackslash}G$.
Now the following holds
\begin{align*}
\opt(G_{PLB})&=\opt(G)+\opt(G_{PLB}\text{\textbackslash}G)\\
&\leq \opt(G)+\frac{|b_1|-n}{2}+\sum_{i=2}^{\log \Delta}\frac{|b_i|}{2}\\
&\leq \opt(G)+\frac{|b_1|-n}{2}\frac{N-|b_1|}{2}\\
&= \opt(G)+\frac{N-n}{2}\\
&\leq \opt(G)+2\cdot(c-1) \opt(G)\\
&=(2c-1)\opt(G),
\end{align*}
where we used our upper bounds on the size of optimal solutions in cycles and $\vec{d}$-Regular cycles in line~2.
That is, $C=2c-1$ in the context of \defref{embedding} and \lemref{embedding}.
Due to \thmref{cubic-MIS} it also holds that $\varepsilon=\frac{140}{139}-\gamma$ for any $\gamma>0$ in the context of \lemref{embedding}.
This gives us an approximation hardness of
\begin{align*}
1+\frac{\varepsilon-1}{(C-1)\varepsilon+1}
&= 1+\frac{\frac{1}{139}-\gamma}{2(c-1)\left(\frac{140}{139}-\gamma\right)+1}\\
& =1+\frac{\frac{1}{139}-\gamma}{2\frac{2c_2\left(\frac{1}{t+1}+\frac{1}{\beta-1}\right)}{1-2c_2\left(\frac{1}{t+1}+\frac{1}{\beta-1}\right)}\left(\frac{140}{139}-\gamma\right)+1}\\
& =1+\frac{\left(\frac{1}{139}-\gamma\right)\left(1-2c_2\left(\frac{1}{t+1}+\frac{1}{\beta-1}\right)\right)}{4c_2\left(\frac{1}{t+1}+\frac{1}{\beta-1}\right)\left(\frac{140}{139}-\gamma\right)+1-2c_2\left(\frac{1}{t+1}+\frac{1}{\beta-1}\right)}
\end{align*}
due to our choice of $c$ in \lemref{cubictoPLB}.
\end{proof}

\subsubsection{Vertex Cover}
\citet{SDY12} showed that an MVC can be found in polynomial time in any $\vec{d}$-Regular cycle.
The same holds true for cycles.
The size of an MVC is $\left\lceil \tfrac{n}2\right\rceil$ in the latter.
For $\vec{d}$-Regular cycles the following lemma gives the size of an MVC.
For the sake of brevity the proof of the lemma is omitted.
\begin{lemma}\label{lem:d-reg-vc}
An $RC_n^{\vec{d}}$ has a minimum vertex cover of size $\frac{|V(RC_n^{\vec{d}})|}{2}$ if $n$ is even and of size $\frac{|V(RC_n^{\vec{d}})|}{2}+1$ if $n$ is odd.
\end{lemma}

Again, we need the following inapproximability result to use the framework.
\begin{theorem}[\cite{DinurSafra05, feige2003vertex}]\label{thm:cubic-MVC}
In regular graphs MVC is hard to approximate within a factor of $10\sqrt{5}-21\approx 1.3606$ unless $\P=\NP$.
\end{theorem}

We can now state the approximation hardness result.
\begin{restatable}{theorem}{stateapxvc}\label{thm:vc-apx}
For every $\beta>1$ and every $t\ge0$ \textsc{Minimum Vertex Cover} cannot be approximated to within a factor of $1+\frac{\left(10\sqrt{5}-22\right)\left(1-2c_2\left(\frac{1}{t+1}+\frac{1}{\beta-1}\right)\right)}{3-4c_2\left(\frac{1}{t+1}+\frac{1}{\beta-1}\right)}$ on graphs with \PLBU, \PLBL and \PLBN unless $\P=\NP$.
\end{restatable}
\begin{proof}
As before, we embed a cubic graph $G$ into a graph $G_{PLB}$ which fulfills \PLBU, \PLBL and \PLBN using \lemref{cubictoPLB}. 
Let $\opt(G)$ and $\opt(G_{PLB})$ denote the size of a minimum vertex cover of $G$ and $G_{PLB}$ respectively. 
We know that $\opt(G)\geq \frac{n}{2}$ for cubic graphs and from \lemref{d-reg-vc} we know the size of the optimal solution for $G_{PLB}\text{\textbackslash}G$.
Now it holds that
\begin{align*}
\opt(G_{PLB})&=\opt(G)+\opt(G_{PLB}\text{\textbackslash}G)\\
&\leq \opt(G)+\left\lceil \frac{|b_1|-n}{2}\right\rceil+\sum_{i=2}^{\log \Delta}\left(\frac{|b_i|}{2}+1\right) \\
&\leq \opt(G)+\frac{|b_1|-n}{2}+1+\frac{N-|b_1|}{2}+\log N \\
&\leq \opt(G)+\frac{N-n}{2}+\log N+1\\
&\leq \opt(G)+\frac{N-n}{2}+n\\
&\leq \opt(G)+\frac{(c+1)n}{2}\\
&\le (c+2)\opt(G),
\end{align*}
where in line~2 we used our upper bounds on the size of optimal solutions in cycles and $\vec{d}$-Regular cycles.
That is, $C=c+2$ in the context of \defref{embedding} and \lemref{embedding}.
Due to \thmref{cubic-MVC} it also holds that $\varepsilon=10\sqrt{5}-21$ in the context of \lemref{embedding}.
This gives us an approximation hardness of
\begin{align*}
1+\frac{\varepsilon-1}{C}
& =1+\frac{10\sqrt{5}-22}{C}\\
& =1+\frac{10\sqrt{5}-22}{c+2}\\
& =1+\frac{10\sqrt{5}-22}{3+\frac{2c_2\left(\frac{1}{t+1}+\frac{1}{\beta-1}\right)}{1-2c_2\left(\frac{1}{t+1}+\frac{1}{\beta-1}\right)}}\\
& =1+\frac{\left(10\sqrt{5}-22\right)\left(1-2c_2\left(\frac{1}{t+1}+\frac{1}{\beta-1}\right)\right)}{3-4c_2\left(\frac{1}{t+1}+\frac{1}{\beta-1}\right)}
\end{align*}
due to our choice of $c$ in \lemref{cubictoPLB}.
\end{proof}

\subsection{Approximation Hardness for Simple Graphs}

In this section we give an embedding of cubic graphs into simple graphs with the \PLBU, \PLBL and \PLBN property. 
We use stars as the gadget for our embeddings into simple graphs.
The following is a simple observation and is therefore stated without a formal proof.

\begin{lemma}\label{lem:star}
A star of size $n$ has a minimum dominating set and a minimum vertex cover of size $1$ and maximum independent set of size $n-1$. 
Also, these can be computed in polynomial time. 
\end{lemma}

Now we show the embedding of cubic graphs into simple graphs with \PLBU, \PLBL and \PLBN. 
Then we can use \lemref{embedding} again to show inapproximability.

\begin{restatable}{lemma}{statecubictoSimple}\label{lem:cubictoSimple}
Any cubic graph $G$ can be embedded into a simple graph $G_{PLB}$ having the \PLBU, \PLBL and \PLBN properties for any $\beta>2$ and any $t\ge 0$.
\end{restatable}  
\begin{proof}
Suppose we are given $\beta$ and $t$. 
Again, we want to determine $c_1$ and $c_2$ of \PLBU and \PLBL respectively.
Let $n$ be the number of nodes in graph $G$ and let $N=cn$ be the number of nodes in $G_{PLB}$ for some constant $c$ to be determined.
Like in \lemref{embedding} we have to ensure a number of conditions to get a valid degree sequence.
To hide a cubic graph in the respective bucket of $G_{PLB}$, we need
\begin{equation*}
c_1N(t+1)^{\beta-1}\sum_{i=2}^{3}(i+t)^{-\beta} = c_1N(t+1)^{\beta-1}\left(\frac{1}{(2+t)^\beta}+\frac{1}{(3+t)^\beta}\right) \ge n.
\end{equation*}
As we will see, we can choose the constant $c_1$ arbitrarily large, so the former condition is no real restriction.
Then we choose the maximum degree $\Delta$ such that
\[d_{max(G_{PLB})}=(cn)^{\frac{1}{\beta-1}}\]
and $d_{min}=1$.

In our embedding we just add for each bucket $d\ge2$ the number of stars of size $2^d+1$ they need to reach their respective lower bounds.
Bucket $1$ can get up to $n$ nodes, since we hide the graph $G$ in it and bucket $0$ gets all the degree-one nodes of our star gadgets.
By filling a bucket (other than buckets $0$ and $1$) we might deviate by at most one from the lower bound of that bucket.
Then, we add additional stars within the bounds of our buckets until we have exactly $N$ nodes.
If we only need one more node, we just add it and connect it to an arbitrary star.
This does not change the properties of the star or the degree of its center enough to make it change its bucket.
In order for this to be possible we need to ensure that after filling all buckets to their lower bound, there is still some slack until we reach $N$.
This is the case if the following inequality holds true
\begin{align*}
& n+\sum_{i=0}^{\left\lfloor \log \Delta\right\rfloor}\hspace{-0.5 ex}\left((2^i+1)\left(1+c_2N(t+1)^{\beta-1}\sum_{j=2^i}^{2^{i+1}-1}(j+t)^{-\beta}\right)\right)\\
& \le \frac{N}{c} + \log N^{\frac1{\beta-1}} + \frac{c_2}{(t+1)}N + \frac{c_2}{\beta-1}N + 2\Delta + c_2 N + \frac{c_2}{\beta-2}N(t+1)\\
& \le N\left(\frac1c+\eta+\frac{c_2}{t+1}+\frac{c_2}{\beta-1}+\eta+c_2+\frac{c_2}{\beta-2}(t+1)\right)\\
& \le N,
\end{align*}
where in the second line we used the inequalities
\begin{align*}
\sum_{i=0}^{\left\lfloor \log \Delta\right\rfloor}\hspace{-0.5 ex}\left(1+c_2N(t+1)^{\beta-1}\sum_{j=2^i}^{2^{i+1}-1}(j+t)^{-\beta}\right) & \le \log N^{\frac1{\beta-1}} + \frac{c_2}{(t+1)}N + \frac{c_2}{\beta-1}N, \\
\sum_{i=0}^{\left\lfloor \log \Delta\right\rfloor}\hspace{-0.5 ex}\left(2^i\left(1+c_2N(t+1)^{\beta-1}\sum_{j=2^i}^{2^{i+1}-1}(j+t)^{-\beta}\right)\right) & \le 2\Delta + c_2 N + \frac{c_2}{\beta-2}N(t+1)
\end{align*}
and choose a constant $\eta>0$ arbitrarily small. 

From this last condition we can derive
\[c\ge1+\frac{\eta'+c_2\left(\frac{1}{t+1}+\frac{1}{\beta-1}+\frac{t+1}{\beta-2}+1\right)}{1-\eta'-c_2\left(\frac{1}{t+1}+\frac{1}{\beta-1}+\frac{t+1}{\beta-2}+1\right)},\]
since $\eta$ and therefore $\eta'$ can be arbitrarily small.
We choose $\eta'=c_2\left(\frac{1}{t+1}+\frac{1}{\beta-1}+\frac{t+1}{\beta-2}+1\right)$ to get
\[c=1+\frac{2c_2\left(\frac{1}{t+1}+\frac{1}{\beta-1}+\frac{t+1}{\beta-2}+1\right)}{1-2c_2\left(\frac{1}{t+1}+\frac{1}{\beta-1}+\frac{t+1}{\beta-2}+1\right)}=\frac1{1-2c_2\left(\frac{1}{t+1}+\frac{1}{\beta-1}+\frac{t+1}{\beta-2}+1\right)}.\]
The use of star gadets means we also have to guarantee that $c_1$ is big enough for all degree-one nodes to fit into bucket $0$. 
Since $c_1$ can be arbitrarily large, this is no problem.

Now we can essentially choose $c_1$ arbitrarily large and $c_2$ arbitrarily small, guaranteeing $c>1$ and a large enough gap to have a valid degree sequence.
At the same time our choice of $c$ guarantees that we can fill the graph with exactly $N$ nodes.
Furthermore, since every node has a constant number of neighbors of equal or higher degree, $G_{PLB}$ also fulfills \PLBN, which always allows us at least $c_3 \log N$ many neighbors.
\end{proof}

The following subsections employ the same framework as for the multigraphs.
Since we already stated the inapproximability results for cubic graphs and the optimal solution sizes for the star gadgets, we can prove the main results immediately.

\subsection{Dominating Set}

\begin{theorem}\label{thm:ds-apx-simple}
For every $\beta>2$ and every $t\ge0$ \textsc{Minimum Dominating Set} cannot be approximated to within a factor of $1+\left(130\cdot\left(4\frac{1-\frac{c_2}{t+1}}{1-2c_2\left(\frac{1}{t+1}+\frac{1}{\beta-1}+\frac{t+1}{\beta-2}+1\right)}+1\right)\right)^{-1}$ on simple graphs with \PLBU, \PLBL and \PLBN unless $\P=\NP$.
\end{theorem}
\begin{proof} 
\lemref{cubictoSimple} gives us an $L$-reduction from a cubic graph $G$ to a simple graph $G_{PLB}$ with the \PLBU, \PLBL and \PLBN properties.
Let $\opt(G)$ and $\opt(G_{PLB})$ denote the size of a minimum dominating set for $G$ and $G_{PLB}$ respectively.
Let $b_i$ be the set of nodes in PLB bucket $i$, i.e. the set of nodes $v\in V$ with $\deg(v)\in\left[2^i,2^{i+1}-1\right]$.
We know that $\opt(G)\geq \frac{n}{4}$ and from \lemref{star} we can derive $\opt\left(G_{PLB}\text{\textbackslash}G\right)=N-n-|b_0|$.
It now holds that
\begin{align*}
\opt(G_{PLB})&=\opt(G)+\opt(G_{PLB}\text{\textbackslash}G)\\
&= \opt(G)+N-n-|b_0|\\
&\leq \opt(G)+N-n-\frac{c_2}{t+1}N\\
&= \opt(G)+\left(c-1-c\frac{c_2}{t+1}\right)n\\
&\leq \opt(G)+\left(c-1-c\frac{c_2}{t+1}\right)4\opt(G)\\
&= \left(4c\left(1-\frac{c_2}{t+1}\right)-3\right)\opt(G).
\end{align*}
In the context of \defref{embedding} and \lemref{embedding} this means $C=4c\left(1-\frac{c_2}{t+1}\right)-3$.
Due to \thmref{cubic-MDS} it also holds that $\varepsilon=\frac{391}{390}$ in the context of \lemref{embedding}.
This gives us an approximation hardness of
\begin{align*}
1+\frac{\varepsilon-1}{C}
	&=1+\frac{3}{390\cdot(4c\left(1-\frac{c_2}{t+1}\right)-3)}\\
	& =1+\frac{3}{390\cdot\left(4\frac{1-\frac{c_2}{t+1}}{1-2c_2\left(\frac{1}{t+1}+\frac{1}{\beta-1}+\frac{t+1}{\beta-2}+1\right)}+1\right)}\\
	& =1+\left(130\cdot\left(4\frac{1-\frac{c_2}{t+1}}{1-2c_2\left(\frac{1}{t+1}+\frac{1}{\beta-1}+\frac{t+1}{\beta-2}+1\right)}+1\right)\right)^{-1}
\end{align*}
due to our choice of $c$ in \lemref{cubictoSimple}.
\end{proof}

\subsection{Independent Set}

\begin{theorem}\label{thm:is-apx-simple}
For every $\beta>2$ and every $t\ge0$ \textsc{Maximum Independent Set} cannot be approximated to within a factor of $1+\frac{\left(\frac{1}{139}-\gamma\right)\left((t+1)\left(1-2c_2\left(\frac{1}{t+1}+\frac{1}{\beta-1}+\frac{t+1}{\beta-2}+1\right)\right)\right)}{4c_1\left(\frac{140}{139}-\gamma\right)+(t+1)\left(1-2c_2\left(\frac{1}{t+1}+\frac{1}{\beta-1}+\frac{t+1}{\beta-2}+1\right)\right)}$ for any $\gamma>0$ on simple graphs with \PLBU, \PLBL and \PLBN unless $\P=\NP$.
\end{theorem}
\begin{proof}
As before we embed a cubic graph $G$ into a graph $G_{PLB}$ which fulfills \PLBU, \PLBL and \PLBN using \lemref{cubictoSimple}. 
Let $\opt(G)$ and $\opt(G_{PLB})$ denote the size of a maximum independent set of $G$ and $G_{PLB}$ respectively. 
We know that $\opt(G)\geq \frac{n}{4}$ and from \lemref{star} we can derive $\opt\left(G_{PLB}\text{\textbackslash}G\right)\le |b_1|$.
Now the following holds
\begin{align*}
\opt(G_{PLB})&=\opt(G)+|b_1|\\
&\leq \opt(G)+\tfrac{c_1}{t+1}N\\
&= \opt(G)+c\tfrac{c_1}{t+1}n\\
&\leq \opt(G)+c\tfrac{c_1}{t+1}4\opt(G)\\
&= \left(4c\tfrac{c_1}{t+1}+1\right)\opt(G).
\end{align*}
This results in $C=4c\tfrac{c_1}{t+1}+1$ in the context of \defref{embedding} and \lemref{embedding}.
Due to \thmref{cubic-MIS} it also holds that $\varepsilon=\frac{140}{139}-\gamma$ for any $\gamma>0$ in the context of \lemref{embedding}.
This gives us an approximation hardness of
\begin{align*}
1+\frac{\varepsilon-1}{(C-1)\varepsilon+1}
&= 1+\frac{\frac{1}{139}-\gamma}{4c\tfrac{c_1}{t+1}\left(\frac{140}{139}-\gamma\right)+1}\\
&= 1+\frac{\frac{1}{139}-\gamma}{\frac{4c_1\left(\frac{140}{139}-\gamma\right)}{(t+1)\left(1-2c_2\left(\frac{1}{t+1}+\frac{1}{\beta-1}+\frac{t+1}{\beta-2}+1\right)\right)}+1}\\
&= 1+\frac{\left(\frac{1}{139}-\gamma\right)\left((t+1)\left(1-2c_2\left(\frac{1}{t+1}+\frac{1}{\beta-1}+\frac{t+1}{\beta-2}+1\right)\right)\right)}{4c_1\left(\frac{140}{139}-\gamma\right)+(t+1)\left(1-2c_2\left(\frac{1}{t+1}+\frac{1}{\beta-1}+\frac{t+1}{\beta-2}+1\right)\right)}
\end{align*}
due to our choice of $c$ in \lemref{cubictoSimple}.
\end{proof}

\subsection{Vertex Cover}

\begin{theorem}\label{thm:vc-apx-simple}
For every $\beta>2$ and every $t\ge0$ \textsc{Minimum Vertex Cover} cannot be approximated to within a factor of $1+\frac{\left({1-2c_2\left(\frac{1}{t+1}+\frac{1}{\beta-1}+\frac{t+1}{\beta-2}+1\right)}\right)\left(10\sqrt{5}-22\right)}{2c_2\left(\frac{1}{\beta-1}+\frac{t+1}{\beta-2}+1\right)+1}$ on simple graphs with \PLBU, \PLBL and \PLBN unless $\P=\NP$.
\end{theorem}
\begin{proof}
We use \lemref{cubictoSimple} to embed a cubic graph $G$ into a graph $G_{PLB}$ which fulfills \PLBU, \PLBL and \PLBN. 
Let $\opt(G)$ and $\opt(G_{PLB})$ denote the size of a minimum vertex cover of $G$ and $G_{PLB}$ respectively. 
We know that $\opt(G)\geq \frac{n}{2}$ for cubic graphs and from \lemref{star} we can derive $\opt\left(G_{PLB}\text{\textbackslash}G\right)=N-n-|b_0|$.
It now holds that
\begin{align*}
\opt(G_{PLB})&=\opt(G)+\opt(G_{PLB}\text{\textbackslash}G)\\
&= \opt(G)+N-n-|b_0|\\
&\leq \opt(G)+N-n-\frac{c_2}{t+1}N\\
&= \opt(G)+\left(c-1-c\frac{c_2}{t+1}\right)n\\
&\leq \opt(G)+\left(c-1-c\frac{c_2}{t+1}\right)2\opt(G)\\
&= \left(2c\left(1-\frac{c_2}{t+1}\right)-1\right)\opt(G),
\end{align*}
which gives us $C=2c\left(1-\frac{c_2}{t+1}\right)-1$ in the context of \defref{embedding} and \lemref{embedding}.
Due to \thmref{cubic-MVC} it also holds that $10\sqrt{5}-21$ in the context of \lemref{embedding}.
This gives us an approximation hardness of
\begin{align*}
1+\frac{\varepsilon-1}{C}
& =1+\frac{10\sqrt{5}-22}{C}\\
& =1+\frac{10\sqrt{5}-22}{2\frac{\left(1-\frac{c_2}{t+1}\right)}{1-2c_2\left(\frac{1}{t+1}+\frac{1}{\beta-1}+\frac{t+1}{\beta-2}+1\right)}-1}\\
& =1+\frac{\left({1-2c_2\left(\frac{1}{t+1}+\frac{1}{\beta-1}+\frac{t+1}{\beta-2}+1\right)}\right)\left(10\sqrt{5}-22\right)}{2\left(1-\frac{c_2}{t+1}\right)+2c_2\left(\frac{1}{t+1}+\frac{1}{\beta-1}+\frac{t+1}{\beta-2}+1\right)-1}\\
& =1+\frac{\left({1-2c_2\left(\frac{1}{t+1}+\frac{1}{\beta-1}+\frac{t+1}{\beta-2}+1\right)}\right)\left(10\sqrt{5}-22\right)}{2c_2\left(\frac{1}{\beta-1}+\frac{t+1}{\beta-2}+1\right)+1}
\end{align*}
due to our choice of $c$ in \lemref{cubictoSimple}.
\end{proof}

\bibliography{PLBCL}
\end{document}